\documentclass[draftclsnofoot, onecolumn, 12pt]{IEEEtran}
\usepackage[utf8]{inputenc}
\usepackage{amsfonts}
\usepackage{amssymb}
\usepackage{stfloats}
\usepackage{graphicx}
\usepackage{amsmath}
\usepackage{array}
\usepackage{epstopdf}
\usepackage{graphicx}
\usepackage{amsthm}
\usepackage{algorithm}
\usepackage[noend]{algpseudocode}
\usepackage{amsmath}
\usepackage{graphics}
\usepackage{epsfig}

\usepackage{bm}
\usepackage{upgreek}
\usepackage{caption}
\usepackage{subcaption}
\usepackage{tabularx}
\usepackage{cite}

\usepackage[usenames]{color}

\newcommand{\hermconj}{^{\mathsf{H}}}
\newcommand{\trans}{^{\mathsf{T}}}
\newcommand{\re}{\operatorname{Re}}

\newcommand{\st}{\operatorname{s.t.}}
\newcommand{\diag}{\operatorname{diag}}
\DeclareMathOperator{\Tr}{Tr}

\DeclareMathOperator{\find}{find}
\DeclareMathOperator{\argmin}{argmin}

\usepackage{mathtools}

\DeclarePairedDelimiter\floor{\lfloor}{\rfloor}

\newtheorem{Lemma}{Lemma}

\newtheorem{Proposition}{Proposition}
\newtheorem{Remark}{Remark}

\begin{document}

\title{Reconfigurable Intelligent Surface Assisted Mobile Edge Computing with Heterogeneous Learning Tasks}

\author{
    Shanfeng Huang, Shuai Wang, Rui Wang, Miaowen Wen, and Kaibin Huang

    \thanks{	
    S. Huang is with the Department of Electrical and Electronic Engineering, Southern University of Science and Technology, Shenzhen 518055, China, and also with the Department of Electrical and Electronic Engineering, The University of Hong Kong, Hong Kong (e-mail: sfhuang@eee.hku.hk).

    S. Wang is with the Department of Electrical and Electronic Engineering, and also with the Department of Computer Science and Engineering, and also with the Sifakis Research Institute for Trustworthy Autonomous Systems, Southern University of Science and Technology, Shenzhen 518055, China (e-mail: wangs3@sustech.edu.cn).
	
    R. Wang is with the Department of Electrical and Electronic Engineering, Southern University of Science and Technology, Shenzhen, China, and also with the Research Center of Networks and Communications, Peng Cheng Laboratory, Shenzhen, China (e-mail: wang.r@sustech.edu.cn).

    M. Wen is with the School of Electronic and Information Engineering, South China University of Technology, Guangzhou 510641, China (e-mail: eemwwen@scut.edu.cn).

	K. Huang is with the Department of Electrical and Electronic Engineering, The University of Hong Kong, Hong Kong (e-mail: huangkb@eee.hku.hk).}
}

\maketitle

\begin{abstract}
The ever-growing popularity and rapid improving of artificial intelligence (AI) have raised rethinking on the evolution of wireless networks.  Mobile edge computing (MEC) provides a natural platform for AI applications since it is with rich computation resources to train machine learning (ML) models, as well as low-latency access to the data generated by mobile and internet of things (IoT) devices. In this paper, we present an infrastructure to perform ML tasks at an MEC server with the assistance of a reconfigurable intelligent surface (RIS). In contrast to conventional communication systems where the principal criterions are to maximize the throughput, we aim at maximizing the learning performance. Specifically, we minimize the maximum learning error of all participating users by jointly optimizing transmit power of mobile users, beamforming vectors of the base station (BS), and the phase-shift matrix of the RIS. An alternating optimization (AO)-based framework is proposed to optimize the three terms iteratively, where a successive convex approximation (SCA)-based algorithm is developed to solve the power allocation problem, closed-form expressions of the beamforming vectors are derived, and an alternating direction method of multipliers (ADMM)-based algorithm is designed together with an error level searching (ELS) framework to effectively solve the challenging nonconvex optimization problem of the phase-shift matrix. Simulation results demonstrate significant gains of deploying an RIS and validate the advantages of our proposed algorithms over various benchmarks. Lastly, a unified communication-training-inference platform is developed based on the CARLA platform and the SECOND network, and a use case (3D object detection in autonomous driving) for the proposed scheme is demonstrated on the developed platform.
\end{abstract}


\IEEEpeerreviewmaketitle

\section{Introduction}
The prevalence of mobile terminals and rapid growth of Internet of Things (IoT) technology have boosted a wide spectrum of new applications, many of which are computation-intensive and latency-critical, such as image recognition, mobile augmented reality, and edge machine intelligence. Mobile edge computing (MEC) is envisioned as a promising paradigm to ease the conflict between resource-hungry applications and resource-limited mobile devices, by providing cloud-computing capabilities within the radio access network in close proximity to mobile subscribers \cite{Mao2017MECsurvey}.

MEC is naturally well-suited for the AI-oriented networks, and  the marriage of MEC and AI has given rise to a new research area, called  ``edge intelligence (EI)'' or ``edge AI'' \cite{Zhou2019EI,Li2019EdgeAI,Zhu2020ELsurvey,Yu2020IE}. In general, there are two ways to realize the vision of edge AI, i.e., model sharing and data sharing \cite{Zhou2019EI,Chen2019DLwithEdgeComp,Wang2020EdgeLearningTWC}. Model sharing is typically achieved by federated learning which jointly exploits  on-device training and federated aggregation, and a series of outstanding works focus on this  type of edge learning \cite{Wang2019FL,Gunduz2019MLintheair,Zhu2019BroadbandAggr,Sun2020FLIoT,Tran2019FL,Du2020SGquant}. However, running computation-intensive
algorithms such as deep neural network models locally is very resource-demanding and requires high-end processors to be
armed in the devices \cite{Zhou2019EI}.  Moreover, training neural network models requires the training data to have both input data and labels. In practice, however, the raw data collected by IoT devices, are generally unlabeled data and cannot be directly used for training.
Therefore, we focus on data sharing where the data collected from the mobile devices (MDs) are offloaded to the MEC server for model training.

While MEC brings many benefits, the last-mile communications from mobile terminals to MEC servers are typically via wireless channels which are prone to channel hostilities. To this end, an emerging paradigm called reconfigurable intelligent surface (RIS) was proposed, aiming at creating a smart radio environment by turning the wireless environment into an optimization variable, which can be controlled and programmed\cite{Renzo2020SREJsac}. Specifically, an RIS is a planar array consisting of a large number of low-cost passive reflecting elements with reconfigurable phase shifts, each of which can be dynamically tuned via a software controller to reflect the incident signals \cite{Basar2019RIS}. Thus, the constructive signals can be combined and the interference can be suppressed by tuning the phase shifts of reflecting elements adaptively \cite{Liaskos2018SCM,Renzo2019SRE,Liu2020RISchannelestimate}. It has been demonstrated that the energy efficiency and throughput can be remarkably enhanced in various wireless communication systems by introducing RISs and jointly optimizing the beamforming vectors and the phase-shift matrices \cite{QWu2019IRSBeamforming,Lin2020risarxiv,Li2020RISUAV,Liu2020RISofdm,Nadeem2020RIS,Han2020IRSpowctrl,Yang2020IRSnoma,Fu2019IRSnoma,Wu2020IRSswip,Guo2020sumrateRIS,Zhang2020LimitedPhase,Di2020HybBF}.

\subsection{Motivations and Related Works}
As elaborated above, MEC provides a perfect platform to house AI applications, and RIS turns the unstable wireless channels into a controllable smart radio environment. In this paper, we investigate the design of an RIS-assisted MEC system with ML tasks. In contrast with conventional communication systems where the  general goals are to maximize the throughput, edge ML systems aim at maximizing the learning performance. As a result, the well-known resource allocation schemes that are optimized for conventional communication systems, such as water-filling scheme \cite{Yu2004waterfilling} and max-min fairness scheme \cite{Li2015mmfairness} may lead to poor learning performance since they do not take into account the learning-specific factors such as model and data complexities. For instance, with the same amount of training data samples, a support vector machine (SVM) and a convolutional neural network (CNN) can achieve different learning accuracies. Moreover, the communication costs for transmitting one data sample for different ML tasks may vary significantly.

Recently, there are some outstanding works that aim to optimize the resource allocation schemes for learning-centric systems. In \cite{Liu2020dataimportance}, the authors proposed a data-importance aware user scheduling scheme for edge ML systems, where data are regarded as having different importance levels based on certain importance measurement and more resources are given to the data with high importance. Nevertheless, the analysis is mainly based on SVM. For more general ML models, the importance of training data is hard to quantify. In \cite{Shi2019RISEL}, the authors investigated an RIS-assisted edge inference system, where the  inference tasks allocation strategy, downlink transmit beamforming and phase shift of the RIS are jointly optimized. In this paper, however, the inference tasks are considered as general edge computing tasks in essence, leading to few insights for real ML tasks.  Besides, the authors in \cite{Yang2020MEC-ML} proposed an MEC-based hierarchical ML tasks distribution framework for industrial IoT, and solved the delay minimization problem considering the ML model complexity and inference error rate, data quality, computing capability at the device and MEC server, and communications bandwidth. More recently, the authors in \cite{Wang2020LearningCentricICC,Wang2020EdgeLearningTWC}  put forth and validated a nonlinear classification error model for ML tasks, based on which a learning-centric power allocation scheme was proposed and shown to outperform conventional water-filling and max-min fairness schemes significantly with respect to learning error. In this paper, we further extend \cite{Wang2020EdgeLearningTWC} to the scenario where an RIS is deployed to provide intelligence to the wireless channels. With the presence of the RIS, the power allocation scheme needs to be redesigned and new challenges in the beamforming vector design and phase shift optimization arise.

\subsection{Our Contributions}
In this paper, we make an attempt on shedding some light on the design of RIS-assisted edge ML with heterogeneous learning tasks. Specifically, we adopt the nonlinear learning error model \cite{Wang2020EdgeLearningTWC,johnson2018accuracypilotdata}, and aim at minimizing the maximum learning error of all the learning tasks by jointly optimizing the transmit power of the mobile devices, the beamforming vectors at the base station (BS) and the phase shift matrix at the RIS. The optimization problem is highly nonconvex and involves too many optimization variables. To address this challenge, we design an alternating optimization (AO)-based framework to decompose the primal problem and each subproblem is efficiently solved either in closed form or with low-complexity algorithms.
Specifically, a successive convex approximation (SCA)-based algorithm is developed to solve the nonconvex power optimization problem. The optimization of beamforming vectors is shown to be equivalent to maximizing the signal-to-interference-plus-noise ratios (SINRs), and closed-form expressions are derived.
To solve the challenging phase- matrix optimization problem, we propose an error level searching (ELS)-based framework to transform the exponential objective into SINR constraints, and exploit alternating direction method of multipliers (ADMM) to decouple the problem to a set of subproblems that can be solved in a distributed manner.
The main contributions of this paper are summarized as follows.
\begin{itemize}
    \item Instead of maximizing the throughput as in conventional communication systems, we formulate the learning error minimization problem, and show that the learning error scales with the number of RIS elements $M$ by a factor of $(\log_2M)^{-d}$, where $d$ represents the difficulty of the learning task.
    \item By deriving the explicit form of ADMM updates in the considered system, it is found that the proposed ADMM-based RIS design is in fact an iterative procedure including individual phase shift update for each user and global aggregation of different users' phase shifts. This allows an efficient and fast implementation in practical RIS edge learning systems.
    \item Simulations on well-known ML models and public datasets verify the nonlinear learning error model, and demonstrate that our proposed scheme can achieve significantly lower learning error than that of various benchmarks.
    \item Lastly, based on the autonomous driving simulator CARLA and the SECOND (Sparse Embedded CONvolutional Detection) neural network, a unified communication-training-inference platform is developed. The new platform provides vivid virtual environments, streaming and dynamic datasets, and supports high-quality 3D visualization and video generation. A use case (3D object detection in autonomous driving) for the proposed scheme is demonstrated based on the developed platform.
\end{itemize}

The remainder of this paper is organized as follows. The system model is introduced in Section II. The problem formulation is illustrated in Section III. An AO-based optimization framework and the corresponding optimization algorithms for power allocation, beamforming vectors and phase-shift matrix are detailed in Section IV. Numerical simulations and experimental results are presented in Section V, and Section VI concludes this paper.

\emph{Notations:} Italic letters, lowercase and uppercase bold letters represent scalars, vectors, and matrices, respectively. The transpose, conjugate, conjugate transpose, matrix inverse, trace operator and diagonal matrix are denoted as $(\cdot)\trans$, $(\cdot)^*$, $(\cdot)\hermconj$, $(\cdot)^{-1}$, $\Tr(\cdot)$ and $\diag(\cdot)$, respectively. $\mathbf I_N$ denotes the $N\times N$ identity matrix and $\|\cdot\|_p$ denotes the $\ell_p$-norm of a vector. $|\cdot|$ and $\re(\cdot)$ respectively denote the modulus and the real part of a complex number, and $\mathcal{CN}(0,1)$ represents the complex Gaussian distribution with zero mean and unit variance. $\floor x$ is the maximum integer less than or equal to $x$.

\section{System Model}
We consider an edge ML system as shown in Fig. \ref{fig:SystemModel}, where an intelligent edge server attached to a BS with $N$ antennas is serving $K$ single-antenna users, each with a ML task.
The communication is assisted by an RIS, consisting of $M$ passive reflecting elements which could rotate the phase of the incident signal waves.
In particular, the edge server is designated to train $K$ classification models by collecting data observed at the $K$ mobile users. The classification models can be CNNs, SVMs, etc.

\begin{figure}[tb]
    \centering
    \includegraphics[width=150mm]{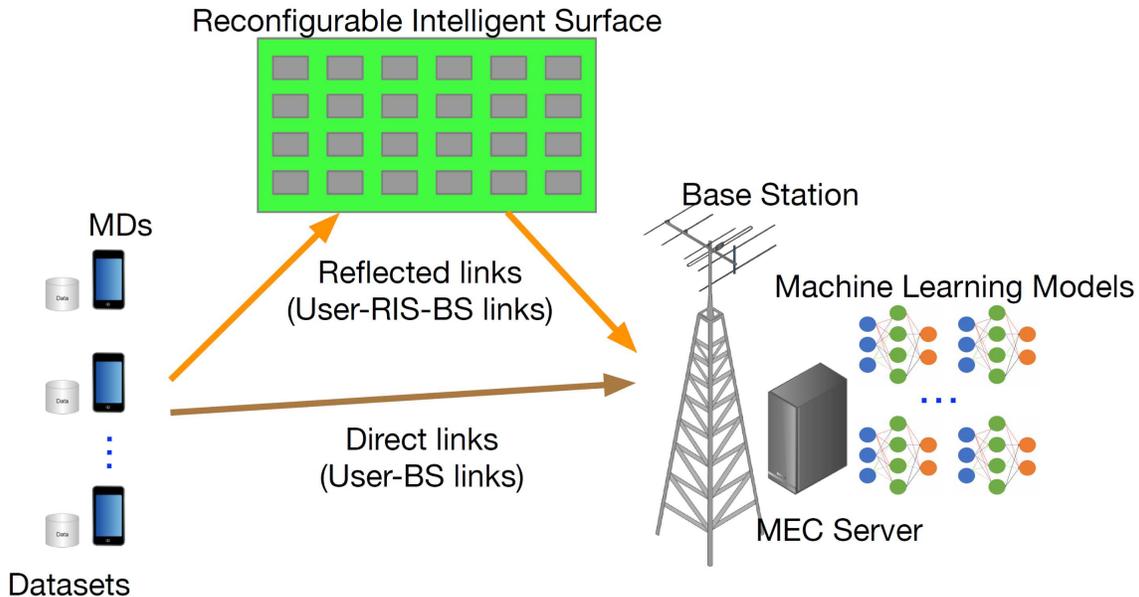}
    \caption{An illustration of the RIS-assisted edge ML system.}
    \label{fig:SystemModel}
\end{figure}

The training data are transmitted from the mobile users to the edge server via wireless channels which have intrinsic random feature due to multi-path effect and  can suffer from high propagation loss \cite{goldsmith_2005}. The wireless channels between the transmitters and receivers may be line-of-sight (LoS) or non-line-of-sight (NLoS). To this end, this paper considers an RIS-assisted scheme that can configure the channel intelligently by tuning the phase shifts of the reflecting elements adaptively. With the presence of the RIS, the channel from user $k$ to the BS includes both the direct link (user-BS link) and the reflected link (user-RIS-BS link), where the direct link includes both LoS and NLoS links, and the reflected link consists of the user-RIS link, the phase shifts at RIS, and the RIS-BS link \cite{QWu2019IRSBeamforming}. Denote the channel vector from $k$-th user to the BS as $\mathbf h_k$. It can be expressed as
\begin{align}
    \mathbf h_k=\underbrace{\mathbf h_{\text{d}, k}}_\text{direct link}+\underbrace{\mathbf G\hermconj \mathbf{\Theta}\hermconj \mathbf h_{\text{r},k}}_\text{reflected link},
\end{align}
where $\mathbf h_{\text{d}, k}\in \mathbb C^{N\times 1}$, $\mathbf h_{\text{r},k}\in \mathbb C^{M\times 1}$, and $\mathbf G\in \mathbb C^{M\times N}$ denote the channel vectors and matrix from user $k$ to the BS, from user $k$ to the RIS, and from the RIS to the BS, respectively. Moreover, $\mathbf{\Theta}=\beta \text{diag}(e^{j\varphi_1},\cdots, e^{j\varphi_M})\in \mathbb C^{M\times M}$ denotes the phase-shift matrix of the RIS, where $\beta\in [0, 1]$ is the amplitude reflection coefficient and $\varphi_{m}\in [0,2\pi)$ is the phase shift of the $i$-th reflecting element. Without loss of generality, $\beta$ is typically set to 1.

Denote the transmitted signal of  user $k \in \{1, 2, \cdots, K\}$ as $x_k$ with power $\mathbb E[|x_k|^2]=p_k$. Accordingly, the received signal $\mathbf y=[y_1, \cdots, y_N]\in \mathbb C^{N\times 1}$ at the BS can be written as
\begin{align}
    \mathbf y=\sum_{k=1}^K \mathbf h_k x_k+\mathbf n,
\end{align}
where $\mathbf n\sim \mathcal{CN}(\mathbf 0, \sigma^2 \mathbf I_N)$ is the additive white Gaussian noise (AWGN) at the BS.
A beamforming vector $\mathbf w_k$ with $\mathbf w_k\hermconj \mathbf w_k=1$ is applied for the received signal from each user $k$. Thus, the estimated symbol at the BS for user $k$ is given by
\begin{align}
    \hat y_k=\mathbf w_k\hermconj \mathbf y=\mathbf w_k\hermconj \mathbf h_k x_k+ \sum_{i=1, i\neq k}^K \mathbf w_k\hermconj \mathbf h_{i}x_{i} +\mathbf w_k\hermconj \mathbf n.
\end{align}
Accordingly, the achievable spectral efficiency of user $k$ in terms of bps/Hz is given by
\begin{align}
    R_k=\log_2\left(1+\frac{p_k |\mathbf w_k\hermconj \mathbf h_k|^2}{\sum_{i=1, i\neq k}^{K}p_{i}|\mathbf w_k\hermconj \mathbf h_{i}|^2+\sigma^2}\right).
\end{align}

Let $B$ denote the bandwidth of the considered system, and $T$ the total transmission time. Thus, the total number of data samples for user $k$'s task is given by
\begin{align}
    v_k=\floor*{\frac{BTR_k}{D_k}}\approx  \frac{BTR_k}{D_k},
\end{align}
where $D_k$ is the number of bits for each data sample, and the approximation is due to $\floor*{x}\to x$ when $x\gg 1$.

\begin{Remark}
In practice, the channel information $\mathbf h_k$,  may not be available at the BS. In such a case, geometry-based ray tracing techniques \cite{He2019raytracing} could be used to estimate the wireless channel.
The geometry can be further obtained by executing the environment sensing task at the edge (e.g., the 3D detection in Section V.D).
This indicates that the proposed system supports not only the ``communication for learning’’ paradigm (e.g., RIS helps the edge to collect more data), but also the ``learning for communication’’ paradigm (e.g., more data helps to better estimate the wireless channel), which is a win-win case for both communication and learning tasks.
\end{Remark}

\section{Problem Formulation}
In contrast with the conventional communication systems where the principal design criterion is usually to maximize the throughput, edge ML systems aim at maximizing the learning performance. Specifically, in the edge ML system considered herein, we aim at optimizing the maximum learning error of all the participating users by jointly optimizing the transmit power $\mathbf p=[p_1, \cdots, p_K]\trans$, the beamforming vectors $\{\mathbf w_k\}_{k=1}^K$ at the BS, and the phase-shift matrix $\mathbf \Theta$ of the RIS. Given the total power budget of users $P$, the vector $\mathbf p$ should satisfy $\sum_{k=1}^K p_k\leq P$. Thus, we have the following optimization problem.
\begin{subequations}
\begin{align}
    \mathcal P: \min_{\mathbf p, \{\mathbf w_k\}_{k=1}^K,\mathbf{\Theta}, \mathbf v} \quad &\max_{k=1,\cdots,K}\quad\Psi_k(v_k) \nonumber\\
    \st \quad\quad &\sum_{k=1}^K p_k\leq P,\  p_k\geq 0, \quad k=1, \cdots, K, \label{eq:pow_constraint}\\
    &\mathbf w_k\hermconj \mathbf w_k=1, \quad k=1,\cdots,K, \label{eq:bf}\\
    &\frac{BTR_k}{D_k}=v_k, \quad k=1, \cdots, K, \label{eq:v_m}\\
    &0\leq \varphi_m<2\pi,\quad m=1, \cdots, M,
\end{align}
\end{subequations}
where $\Psi_k(v_k)$ is the classification error of learning model $k$ given the sample size $v_k$. In general, the functions $\{\Psi_1, \cdots, \Psi_K\}$ can hardly be expressed analytically. Propitiously, their approximate expressions can be obtained based on the analysis in \cite{Wang2020EdgeLearningTWC, johnson2018accuracypilotdata,BELEITES2013samplesize}. Here, we simply adopt the non-linear model developed in \cite{Wang2020EdgeLearningTWC}, i.e.,
\begin{align}\label{eq:errormodel}
    \Psi_k(v_k)\approx c_k v_k^{-d_k},
\end{align}
where $c_k$ and $d_k$ are tuning parameters which can be obtained by curve fitting.

Denote $\bm\theta=[e^{j\varphi_1},\cdots, e^{j\varphi_M}]\trans\in\mathbb C^{M\times 1}$. We have $|\theta_m|=1$ for all $m$. Then, by substituting (\ref{eq:v_m}) and (\ref{eq:errormodel}) into the objective function, problem $\mathcal P$ is transformed into the following problem.
\begin{subequations}
    \begin{align}
        \mathcal P1: \min_{\mathbf p,\{\mathbf w_k\}_{k=1}^K,\mathbf{\Theta}}\max_{k=1,\cdots,K} &c_k\left[\frac{BT}{D_k}\log_2\left(1+\frac{|\mathbf w_k\hermconj(\mathbf h_{\text{d},k}+\mathbf G\hermconj \mathbf{\Theta}\hermconj \mathbf h_{\text{r},k})|^2p_k}{\sum_{i=1, i\neq k}^K |\mathbf w_k\hermconj (\mathbf h_{\text{d},i}+\mathbf G\hermconj \mathbf{\Theta}\hermconj \mathbf h_{\text{r},i})|^2 p_{i}+\sigma^2}\right)\right]^{-d_k}\nonumber\\
        \st \quad\quad &\sum_{k=1}^K p_k\leq P,\ p_k\geq 0,\quad k=1,\cdots, K,\\
        &\mathbf w_k\hermconj \mathbf w_k=1, \quad k=1,\cdots,K, \label{eq:beam}\\
        &|\theta_m|=1,\quad m=1,\cdots,M.
    \end{align}
\end{subequations}

\begin{Remark}[Scaling law of large number of reflecting elements with a single user and single-antenna BS]
To gain some insights on how the number of reflecting elements affect the learning accuracy, we consider the case with a single user and a single-antenna BS, i.e., $K=1$ and $N=1$, and ignore the direct link. Thus, $\mathbf G$ becomes a vector and is denoted by $\mathbf g$.  The received signal-to-noise ratio (SNR) becomes $p|\mathbf h_{\text{r}}\hermconj \mathbf\Theta\mathbf g|/{\sigma}^2$. Assume $\mathbf\Theta=\mathbf I_M$, $\mathbf h_{\text{r}}\sim\mathcal{CN}(\mathbf 0, \varrho_h^2\mathbf I_M)$, and $\mathbf g\sim\mathcal{CN}(\mathbf 0, \varrho_g^2\mathbf I_M)$. According to the central limit theorem, we have $\mathbf h_{\text{r}}\hermconj\mathbf g\sim\mathcal{CN}(\mathbf 0, M\varrho_h^2\varrho_g^2)$ as $M\to\infty$. Thus, the average  received SNR is $p\mathbb|\mathbf h_{\text{r}}\hermconj \mathbf\Theta\mathbf g|/{\sigma}^2\sim Mp\varrho_h^2\varrho_g^2$. This indicates that the learning error is asymptotically proportional to $(\log_2(M))^{-d}$.
\end{Remark}

\section{Joint Power Control and Phase-Shifter Design}
Note that problem $\mathcal P1$ is highly nonconvex due to the nonlinear learning error model in the objective function and the unit-modulus constraints. Moreover, the large number of optimization variables make the problem even more untractable. Fortunately, the optimization of transmit power, beamforming vectors and the phase-shift matrix can be decomposed. Hence, we adopt an AO-based algorithm to solve $\mathcal P1$ in an iterative manner via alternatively optimizing $\mathbf p$, $\{\mathbf w_k\}_{k=1}^K$ and $\mathbf\Theta$.

\subsection{Transmit Power Optimization}
When $\{\mathbf w_k\}_{k=1}^K$ and $\mathbf \Theta$ are fixed, the composite channel gains are known, and there are only the power budget constraints. Thus, the optimization of $\mathbf p$ can be written as
\begin{subequations}\label{prob:p_opt}
    \begin{align}
        \mathcal P_{\mathbf p}: \min_{\mathbf p}\quad\max_{k\in\mathcal K} &\quad \Phi_k(\mathbf p)\\
        \st \quad\quad &\sum_{k=1}^K p_k\leq P,\ p_k\geq 0,\quad k=1,\cdots, K,
    \end{align}
\end{subequations}
where
\begin{align}
    \Phi_k(\mathbf p)=c_k\left[\frac{BT}{D_k}\log_2\left(1+\frac{|\mathbf w_k\hermconj(\mathbf h_{\text{d},k}+\mathbf G\hermconj \mathbf{\Theta}\hermconj \mathbf h_{\text{r},k})|^2p_k}{\sum_{i=1, i\neq k}^K |\mathbf w_k\hermconj(\mathbf h_{\text{d},i}+\mathbf G\hermconj \mathbf{\Theta}\hermconj \mathbf h_{\text{r},i})|^2 p_{i}+\sigma^2}\right)\right]^{-d_k}.
\end{align}
It is observed that the objective function of problem $\mathcal P_{\mathbf p}$ is nonconvex in $\mathbf p$. We adopt the SCA approach to solve problem $\mathcal P_{\mathbf p}$, where a sequence of convex upper bounds \{$\widetilde{\Phi}_k(\mathbf p)$\} are constructed to approximate ${\Phi}_k(\mathbf p)$. Specifically, given any feasible solution $\mathbf p^{\star}$ to $\mathcal P_{\mathbf p}$, we define the surrogate functions
\begin{align}\label{eq:power_bound}
    \widetilde{\Phi}_k(\mathbf p|\mathbf p^{\star})=&c_k\bigg\{\frac{BT}{D_k \ln2}\bigg[\ln\left(\sum_{i=1}^K|\mathbf w_k\hermconj(\mathbf h_{\text{d},i}+\mathbf G\hermconj \mathbf{\Theta}\hermconj \mathbf h_{\text{r},i})|^2 p_{i}+\sigma^2\right)\nonumber\\
    &-\frac{\sum_{i=1,i\neq k}^K|\mathbf w_k\hermconj(\mathbf h_{\text{d},i}+\mathbf G\hermconj \mathbf{\Theta}\hermconj \mathbf h_{\text{r},i})|^2 p_{i}+\sigma^2}{\sum_{i=1,i\neq k}^K|\mathbf w_k\hermconj(\mathbf h_{\text{d},i}+\mathbf G\hermconj \mathbf{\Theta}\hermconj \mathbf h_{\text{r},i})|^2 p_{i}^{\star}+\sigma^2}\nonumber\\
    &-\ln\bigg(\sum_{i=1,i\neq k}^K|\mathbf w_k\hermconj(\mathbf h_{\text{d},i}+\mathbf G\hermconj \mathbf{\Theta}\hermconj \mathbf h_{\text{r},i})|^2 p_{i}^{\star}+\sigma^2\bigg)+1\bigg]\bigg\}^{-d_k},\ k=1,\cdots,K,
\end{align}
which have the following properties.
\begin{Proposition}\label{prop:power}
    The functions $\{\widetilde{\Phi}_k\}$ satisfy the following conditions:\\
    \begin{enumerate}
        \item Upper bound condition: $\widetilde\Phi_k(\mathbf p|\mathbf p^{\star})\geq \Phi_k(\mathbf p)$;
        \item Convexity: $\widetilde\Phi_k(\mathbf p|\mathbf p^{\star})$ is convex in $\mathbf p$;
        \item Local condition: $\widetilde{\Phi}_k(\mathbf p^{\star}|\mathbf p^{\star})=\Phi_k(\mathbf p^{\star})$ and $\nabla_{\mathbf p}\widetilde{\Phi}_k(\mathbf p^{\star}|\mathbf p^{\star})=\nabla_{\mathbf p}\Phi_k(\mathbf p^{\star})$.
    \end{enumerate}
\end{Proposition}
\begin{proof}
    Please refer to Appendix A.
\end{proof}

Replacing the objective function of problem $\mathcal P_{\mathbf p}$ with the surrogate functions and applying the SCA algorithm, we have a sequence of optimization problems.
\begin{subequations}
    \begin{align}
        \mathcal P_{\mathbf p}[n+1]: \min_{\mathbf p}\quad &\max_{k\in\mathcal K}\widetilde\Phi_k(\mathbf p|\mathbf p^{[n]})\nonumber\\
        \st \quad &\sum_{k=1}^Kp_k\leq P,\ p_k\geq 0,\ \forall k=1,\cdots, K,
    \end{align}
\end{subequations}
where $\mathcal P_{\mathbf p}[n+1]$ is the optimization problem in the $(n+1)$-th iteration of the SCA algorithm, and $\mathbf p^{[n]}$ is the optimal solution of $\mathcal P_{\mathbf p}[n]$. Note that each $\mathcal P_{\mathbf p}[n]$ is a convex optimization problem and can be solved efficiently via off-the-shelf toolbox (e.g. CVX). According to Theorem 1 of \cite{BMarks1978InnerApprox}, the sequence $(\mathbf p^{[0]},\mathbf p^{[1]},\cdots)$ converges to the KKT solution to problem $\mathcal P_{\mathbf p}$ for any feasible starting point $\mathbf p^{[0]}$. As a summary, the SCA-based optimization procedure is given in Algorithm \ref{alg:MMpower}.

\begin{algorithm}
\caption{SCA-based algorithm for optimizing $\mathbf p$}\label{alg:MMpower}
\begin{algorithmic}[1]
\State \textbf{Input} $(M,N,K,P,B,T,\mathbf w_k, \mathbf \Theta,\mathbf h_k, ,\sigma^2,c_k,d_k,D_k)$ for $k=1,\cdots,K$.
\State \textbf{Initialize} $\mathbf p^{[0]}=P/K\mathbf 1_K$. Set iteration counter $n=0$.
\State \textbf{Repeat}\\
\quad Update $\mathbf p^{[n+1]}$ by solving $\mathcal P_{\mathbf p}[n+1]$ via CVX.\\
\quad $n\leftarrow n+1.$
\State \textbf{Until} convergence.
\State \textbf{Output} $\mathbf p^{\diamond}=\mathbf p^{[n]}$.
\end{algorithmic}
\end{algorithm}

\subsection{Beamforming Vectors Optimization}

Note that given $\mathbf\Theta$ and $\mathbf p$, the objective function of the original problem $\mathcal P1$ is still nonconvex in $\mathbf w_k$. However, since the objective function is monotonically decreasing in the SINR of each user and $\{\mathbf w_k\}_{k=1}^K$ are decomposable, the optimization of $\mathbf w_k$ with fixed $\mathbf p$ and $\mathbf \Theta$ can be equivalently solved by maximizing the SINR of each user $k$. Consequently, the optimal beamforming vectors can be obtained by solving the following $K$ subproblems.
\begin{subequations}
    \begin{align}
        \mathcal P_{\mathbf w_k}: \max_{\mathbf w_k} \quad &\frac{|\mathbf w_k\hermconj(\mathbf h_{\text{d},k}+\mathbf G\hermconj \mathbf{\Theta}\hermconj \mathbf h_{\text{r},k})|^2p_k}{\sum_{i=1, i\neq k}^K |\mathbf w_k\hermconj (\mathbf h_{\text{d},i}+\mathbf G\hermconj \mathbf{\Theta}\hermconj \mathbf h_{\text{r},i})|^2 p_{i}+\sigma^2}\nonumber\\
        \st\quad &\mathbf w_k\hermconj \mathbf w_k=1.
    \end{align}
\end{subequations}

Although each problem $\mathcal P_{\mathbf w_k}$ is still nonconvex in $\mathbf w_k$, its optimal solution can be achieved in closed-form as given in the following lemma.

\begin{Lemma}\label{lem:opt_bf}
    Given $\mathbf\Theta$ and $\mathbf p$, the optimal solution of $\mathcal P_{\mathbf w_k}$ for arbitrary $k$ is given in closed-form by
    \begin{align}\label{eq:optbeam}
        \mathbf w_k^{\diamond}=\frac{\left(\mathbf I_N+\sum_{i=1}^K\frac{p_{i}}{\sigma^2}\mathbf h_{i}\mathbf h_{i}\hermconj\right)^{-1}\mathbf h_k}{\left\|\left(\mathbf I_N+\sum_{i=1}^K\frac{p_{i}}{\sigma^2}\mathbf h_{i}\mathbf h_{i}\hermconj\right)^{-1}\mathbf h_k\right\|_2},
    \end{align}
    where $\mathbf h_i=\mathbf h_{\text{d},i}+\mathbf G\hermconj \mathbf{\Theta}\hermconj \mathbf h_{\text{r},i}$, for $i=1,\cdots,K$.
\end{Lemma}
\begin{proof}
    Please refer to Appendix B.
\end{proof}

\subsection{Phase-shift Matrix Optimization}
Given transmit power $\mathbf p$ and beamforming vectors $\{\mathbf w_k\}_{k=1}^K$, there remain only the unit-modulus constraints of the RIS elements. By exploiting $\mathbf\Theta=\diag(\bm\theta)$ and setting $\mathbf a_{k,i}=\beta \diag(\mathbf h_{\text{r},i}\hermconj)\mathbf G\mathbf w_k$, $b_{k,i}=\mathbf h_{\text{d},i}\hermconj \mathbf w_k$, the optimization of phase-shift matrix $\mathbf\Theta$ can be equivalently written as the following problem.
\begin{subequations}
    \begin{align}
        \mathcal P_{\bm\theta}:  &\min_{\bm\theta}\quad\max_{k\in\mathcal K} c_k\left[\frac{BT}{D_k}\log_2
        \left(1+\frac{|\bm\theta\hermconj\mathbf a_{k,k}+b_{k,k}|^2p_k}{\sum_{i=1,i\neq k}^K|\bm\theta\hermconj\mathbf a_{k,i}+b_{k,i}|^2p_i+\sigma^2}\right)\right]^{-d_k}\nonumber\\
        &\st \quad |\theta_m|= 1, \forall m=1,\cdots,M.
    \end{align}
\end{subequations}

A common approach to address the nonconvex unit-modulus constraints is semidefinite relaxation (SDR). Nevertheless, even SDR can circumvent the nonconvex unit-modulus constraints, the objective function remains nonconvex due to the nonlinear learning error model. Moreover, the solution achieved by SDR generally does not conform to the rank-1 constraint, and large number of Gaussian randomizations \cite{Luo2010SDR} are required to find a rank-1 solution, which increases the complexity dramatically. Besides, SDR lifts the optimization variable from an $M\times 1$ vector to an $M\times M$ matrix.  Thus, SDR cannot scale up the number of RIS elements. To this end, we propose an ELS framework and an ADMM-based algorithm to solve problem $\mathcal P_{\bm\theta}$. Specifically, we first define the error level of the $k$-th ML task for all $k$ as
\begin{align}
\delta_k=c_k\left[\frac{BT}{D_k}\log_2
\left(1+\frac{|\bm\theta\hermconj\mathbf a_{k,k}+b_{k,k}|^2p_k}{\sum_{i=1,i\neq k}^K|\bm\theta\hermconj\mathbf a_{k,i}+b_{k,i}|^2p_i+\sigma^2}\right)\right]^{-d_k}.
\end{align}
Thus, the maximum error level of all participating tasks is given by $\delta=\max_{k\in\mathcal K} \delta_k$. Then, for a given error level $\delta$, problem $\mathcal P_{\bm\theta}$ can be equivalently transformed to the following feasibility problem.
\begin{subequations}
    \begin{align}
        \mathcal P'_{\bm\theta}:  &\find\quad\bm\theta\\
        &\st\quad\frac{|\bm\theta\hermconj\mathbf a_{k,k}+b_{k,k}|^2p_k}{\sum_{i=1,i\neq k}^K|\bm\theta\hermconj\mathbf a_{k,i}+b_{k,i}|^2p_i+\sigma^2}\geq \gamma_k,\quad k=1,\cdots,K\\
        &\quad\quad\quad |\theta_m|= 1, \quad m=1,\cdots,M,
    \end{align}
\end{subequations}
where $\gamma_k=2^{\frac{D_k\left(\frac{c_k}{\delta}\right)^{\frac{1}{d_k}}}{BT}}-1$. If problem $\mathcal P_{\bm\theta}'$ is feasible, we can reduce $\delta$; otherwise, we increase $\delta$ to make $\mathcal P_{\bm\theta}'$ feasible, until $\delta$ converges to a certain value.

In the sequel, we design an ADMM-based algorithm to solve problem $\mathcal P_{\bm\theta}'$. By introducing a series of auxiliary variables $\{\mathbf q_k\}_{k=1}^K$ and a new constraint $\mathbf q_1=\mathbf q_2=\cdots=\mathbf q_K=\bm\theta$, problem $\mathcal P_{\bm\theta}'$ can be further rewritten as the following form.
\begin{subequations}
    \begin{align}
        \find &\quad\{\mathbf q_k\}_{k=1}^K, \bm\theta\\
        \st &\frac{|\mathbf q_k\hermconj\mathbf a_{k,k}+b_{k,k}|^2p_k}{\sum_{i=1,i\neq k}^K|\mathbf q_k\hermconj\mathbf a_{k,i}+b_{k,i}|^2p_i+\sigma^2}\geq \gamma_k,\quad k=1,\cdots,K\label{eq:snr_constraint}\\
        &|\theta_m|= 1, \quad m=1,\cdots,M\label{eq:norm_constraint}\\
        &\mathbf q_k=\bm\theta, \quad k=1,\cdots,K.
    \end{align}
    \label{prob:theta_admm}
\end{subequations}
The augmented Lagrangian (using the scaled dual variable) of problem (\ref{prob:theta_admm}) is given by
\begin{align}
    \mathcal L_{\rho}(\mathbf q_1,\cdots,\mathbf q_K,\bm\theta,\mathbf u_1,\cdots,\mathbf u_K)=\quad\sum_{k=1}^K \mathbb I_{\mathcal B_k}(\mathbf q_k)+ \mathbb I_{\mathcal C}(\bm\theta)+\rho\sum_{k=1}^K\|\mathbf q_k-\bm\theta+\mathbf u_k\|^2,
\end{align}
where $\mathcal B_k$ is the feasibility region of the $k$-th constraint in (\ref{eq:snr_constraint}) and $\mathcal C$ is the feasibility region of constraint (\ref{eq:norm_constraint}), $\rho>0$ is the penalty parameter, and $\mathbf u_k$ is the scaled dual variable.
Moreover, $\mathbb I$ is the indicator function with
\begin{align}
    \mathbb I_{\mathcal X}(\mathbf x)=
    \begin{cases}
        0, & \text{If} \quad\mathbf x\in\mathcal X,\\
        +\infty, &\text{Otherwise.}
    \end{cases}
\end{align}
The ADMM algorithm iteratively update $\mathbf q_k$, $\bm\theta$ and $\mathbf u_k$ as follows, until a feasible solution is found.
\begin{subequations}\label{eq:admm_update}
    \begin{align}
        &\mathbf q_k^{t+1}:=\argmin_{\mathbf q_k} \mathcal L_{\rho}(\mathbf q_1,\cdots,\mathbf q_K,\bm\theta^t,\mathbf u_1^t,\cdots,\mathbf u_K^t), k=1,\cdots,K\\
        &\bm\theta^{t+1}:=\argmin_{\bm\theta} \mathcal L_{\rho}(\mathbf q_1^{t+1},\cdots,\mathbf q_K^{t+1},\bm\theta,\mathbf u_1^t,\cdots,\mathbf u_K^t)\\
        &\mathbf u_k^{t+1}:=\mathbf u_k^t+\mathbf q_k^{t+1}-\bm\theta^{t+1},k=1,\cdots,K
    \end{align}
\end{subequations}

In the sequel, we show that each update in (\ref{eq:admm_update}) can be efficiently solved either in closed-form or with very low complexity.

1) $\mathbf q_k$ update: The update of $\mathbf q_k$ can be equivalently written as the following problem after removing the irrelevant terms.
\begin{align}\label{eq:q-update}
    \mathbf q_k^{t+1}=\argmin_{\mathbf q_k}\quad\sum_{k=1}^K \mathbb I_{\mathcal A_k}(\mathbf q_k)+\rho\sum_{k=1}^K\|\mathbf q_k-\bm\theta^t+\mathbf u_k^t\|^2.
\end{align}
Note that the update of $\mathbf q_k$ can be decoupled into $K$ subproblems for each $k\in\mathcal K$.
\begin{subequations}
\begin{align}
    \min_{\mathbf q_k} \quad&\|\mathbf q_k-\bm\theta^t+\mathbf u_k^t\|^2\\
    \st \quad&\frac{|\mathbf q_k\hermconj\mathbf a_{k,k}+b_{k,k}|^2p_k}{\sum_{i=1,i\neq k}^K|\mathbf q_k\hermconj\mathbf a_{k,i}+b_{k,i}|^2p_i+\sigma^2}\geq \gamma_k.
\end{align}
\label{prob:q_update}
\end{subequations}

Although problem (\ref{prob:q_update}) is nonconvex in general, strong duality holds and the Lagrangian relaxation produces the optimal solution since there is only one constraint \cite{boyd2004convex}. Thus, we can solve it efficiently using the Lagrangian dual method. Rephrasing problem (\ref{prob:q_update}), it can be equivalently written as the following compact form
\begin{subequations}
    \begin{align}
        \min_{\mathbf q_k} \quad&\|\mathbf q_k-\bm\zeta_k^t\|^2\\
        \st \quad&\mathbf q_k\hermconj \mathbf A_k\mathbf q_k-2\re\{\mathbf b_k\hermconj\mathbf q_k\}= \tau_k,
    \end{align}
    \label{prob:q_qcqp1}
\end{subequations}
where $\bm\zeta_k^t=\bm\theta^t-\mathbf u_k^t$, $\mathbf A_k=\gamma_k\sum_{i=1,i\neq k}^K\mathbf a_{k,i}\mathbf a_{k,i}\hermconj p_i-\mathbf a_{k,k}\mathbf a_{k,k}\hermconj p_k$, $\mathbf b_k=\mathbf a_{k,k}b_{k,k}^* p_k-\gamma_k\sum_{i=1,i\neq k}^K\mathbf a_{k,i}b_{k,i}^*p_i$, and $\tau_k=|b_{k,k}|^2p_k-\gamma_k\sum_{i=1,i\neq k}^K|b_{k,i}|^2p_i-\gamma_k\sigma^2$. Note that we have changed the constraint to equality to simplify the follow-up derivations. When considering the inequality constraint, we can just check whether $\mathbf q_k=\bm\zeta_k^t$ is feasible. If yes, $\mathbf q_k^*=\bm\zeta_k^t$ is the optimal solution; if not, the optimal solution must satisfy the equality constraint.

For ease of notation, we neglect the subscript $k$ in problem (\ref{prob:q_qcqp1}), and let $\mathbf A=\mathbf Q\mathbf\Lambda\mathbf Q\hermconj$ be the eigenvalue decomposition. Then, problem (\ref{prob:q_qcqp1}) is equivalent to
\begin{subequations}
    \begin{align}
        \min_{\tilde{\mathbf q}}\quad&\|\tilde{\mathbf q}-\tilde{\bm\zeta}^t\|^2\\
        \st \quad&\tilde{\mathbf q}\hermconj \mathbf\Lambda\tilde{\mathbf q}-2\re\{\tilde{\mathbf b}\hermconj\tilde{\mathbf q}\}= \tau,
    \end{align}
    \label{prob:q_qcqp1_equ}
\end{subequations}
where $\tilde{\mathbf q}=\mathbf Q\hermconj\mathbf q$, $\tilde{\bm\zeta}^t=\mathbf Q\hermconj\bm\zeta^t$, and $\tilde{\mathbf b}=\mathbf Q\hermconj\mathbf b$.

As a result, the optimal solution can be efficiently found by the following lemma.
\begin{Lemma}\label{lem:opt_q}
    The optimal solution of problem (\ref{prob:q_qcqp1_equ}) is given by
    \begin{align}\label{eq:q-opt}
        \tilde{\mathbf q}^*=(\mathbf I+\mu\mathbf\Lambda)^{-1}(\tilde{\bm\zeta}+\mu\tilde{\mathbf b})
    \end{align}
    where $\mu$ is the Lagrangian multiplier of problem (\ref{prob:q_qcqp1_equ}). Moreover, $\mu$ can be found by solving a nonlinear equation $\chi(\mu)=0$ with
    \begin{align}
        \chi(\mu)=\sum_{m=1}^M\lambda_m\left|\frac{\tilde{\zeta}_m+\mu\tilde b_m}{1+\mu\lambda_m}\right|^2-2\re\left\{\sum_{m=1}^M\tilde b_m^*\frac{\tilde{\zeta}_m+\mu\tilde b_m}{1+\mu\lambda_m}\right\}-\tau,
    \end{align}
    where $\lambda_m$ is the $m$-th diagonal entry of $\mathbf\Lambda$.
\end{Lemma}
\begin{proof}
    Please refer to Appendix C.
\end{proof}

Taking derivative on $\chi(\mu)$ with respect to $\mu$, we have
\begin{align}
    \chi'(\mu)=-2\sum_{m=1}^M\frac{|\tilde b_m-\lambda_m\tilde\zeta_m|^2}{(1+\mu\lambda_m)^3}.
\end{align}
Since we assume the feasibility of problem (\ref{prob:q_qcqp1_equ}), there must exist $\mu$ with $\mathbf I+\mu\mathbf\Lambda\succeq 0$, such that value of $\tilde{\mathbf q}$ minimizing the Lagrangian also satisfies the equality constraint. Thus, $1+\mu\lambda_m\geq 0, \ m=1,\cdots,M$, and $\chi'(\mu)<0$. Therefore, $\chi(\mu)$ is monotonic in the possible region of the solution, and any local solution is guaranteed to be the unique solution. Moreover, the equation $\chi(\mu)=0$ can be efficiently solved by either bisection search method or Newton's method as detailed in Algorithm \ref{alg:Newton}.

\begin{algorithm}
    \caption{Solving $\chi(\mu)=0$ using Newton's method}\label{alg:Newton}
    \begin{algorithmic}[1]
    \State \textbf{Set} $\mu=-(\lambda_{\min}+\lambda_{\max})/2\lambda_{\min}\lambda_{\max}$, $\varepsilon=10^{-6}$.
    \State \textbf{Repeat}\\
    \quad $\mu=\mu-\chi(\mu)/\chi'(\mu)$.
    \State \textbf{Until} $-\chi(\mu)^2/\chi'(\mu)<\varepsilon$.
    \State \textbf{Output} $\mu$.
    \end{algorithmic}
\end{algorithm}

After obtaining $\tilde{\mathbf q}_k$ from problem (\ref{prob:q_qcqp1_equ}), the optimal $\mathbf q_k$ update is given by
\begin{align}
    \mathbf q_k^{t+1}=\mathbf Q\tilde{\mathbf q}_k.
\end{align}

2) $\bm\theta$ update: The update of $\bm\theta$ can be obtained by solving the following problem.
\begin{align}
    \bm\theta^{t+1}&=\argmin_{\bm\theta} \sum_{k=1}^K\|\mathbf q_k^{t+1}-\bm\theta+\mathbf u_k^t\|^2\nonumber\\
    &\st \quad|\theta_m|=1, m=1,\cdots,M.
\end{align}
Thus, the optimal $\bm\theta$ is simply the projection of $\frac{1}{K}\sum_{k=1}^K(\mathbf q_k^{t+1}+\mathbf u_k^t)$ onto the unit-modulus constraints, i.e.,
\begin{align}\label{eq:theta_update}
    \bm\theta^{t+1}=e^{j\angle \frac{1}{K}\sum_{k=1}^K(\mathbf q_k^{t+1}+\mathbf u_k^t)}.
\end{align}

3) $\mathbf u_k$ update: The optimal update of $\mathbf u_k$ can be derived by setting the derivative of the augmented Lagrangian to zero. Thus, we have
\begin{align}\label{eq:u_update}
    \mathbf u_k^{t+1}=\mathbf u_k^t+\mathbf q_k^{t+1}-\bm\theta^{t+1}.
\end{align}

\begin{algorithm}
    \caption{ELS and ADMM-based optimization of phase-shift matrix}\label{alg:ADMM}
    \begin{algorithmic}[1]
    \State \textbf{Set} $\delta_{\min}=0$, $\delta_{\max}=1$, maximum number of iterations $\mathsf{nIter}=1000$, threshold $\epsilon=10^{-6}$, and feasibility indicator $\mathsf{feas}=0$.
    \State \textbf{While} $\delta_{\max}-\delta_{\min}>10^{-4}$\\
    \quad Set $\delta=\frac{\delta_{\min}+\delta_{\max}}{2}$, and substitute current $\delta$ to problem (\ref{prob:theta_admm}).
    \State \quad \textbf{For} $i=1,\cdots,\mathsf{nIter}$\\
    \quad \quad Update $\bm\theta$ via equation (\ref{eq:theta_update}).\\
    \quad \quad\textbf{For} $k=1,\cdots,K$\\
    \quad \quad \quad  Update $\mathbf q_k$ by solving problem (\ref{prob:q_update}).\\
    \quad \quad \quad Update $\mathbf u_k$ via equation (\ref{eq:u_update}).\\
    \quad \quad \textbf{End}\\
    \quad\quad\textbf{If} $\sum_{k=1}^K(\mathbf q_k-\bm\theta)<\epsilon$\\
    \quad\quad\quad $\mathsf{feas}=1$\\
    \quad\quad\quad \textbf{break}\\
    \quad\quad \textbf{End}
    \State \quad\textbf{End} \\
    \quad \textbf{If} $\mathsf{feas}==1$\\
    \quad\quad $\delta_{\max}=\delta$\\
    \quad\textbf{Else}\\
    \quad\quad $\delta_{\min}=\delta$\\
    \quad\textbf{End}\\
    \textbf{End}
    \State \textbf{Output} $\bm\theta$.
    \end{algorithmic}
\end{algorithm}

As a result, the ELS and ADMM-based optimization algorithm for solving problem $\mathcal P_{\bm\theta}$ is summarized in Algorithm \ref{alg:ADMM}.

\begin{Remark}[Distributed implementation]
It can be seen from equations (\ref{eq:q-update})(\ref{eq:theta_update})(\ref{eq:u_update}) that the proposed ADMM-based RIS design is in fact an iterative procedure including local phase shift update at each user and global phase aggregation at the BS. This allows an distributed implementation that supports massive number of user devices. Moreover, the global phase aggregation (30) can be computed using the recent over-the-air-computation (AirComp) technique \cite{Zhu2020AirComp}.
\end{Remark}

\begin{algorithm}
    \caption{The alternating optimization of $\mathcal P1$} \label{alg:AO}
    \begin{algorithmic}[1]
    \State \textbf{Input} $(M,N,K,B,T,\mathbf h_k,\sigma^2,c_k,d_k,D_k)$ for $k=1,\cdots,K$.
    \State \textbf{Initialize} $\mathbf p^0$, $\mathbf w_k^0$ and $\bm \theta^0$. Set iteration counter $t=0$.
    \State \textbf{Repeat}\\
    \quad Update $\mathbf p^{t+1}$ via Algorithm \ref{alg:MMpower}.\\
    \quad Update $\mathbf w_k^{t+1}$ via Equation (\ref{eq:optbeam}).\\
    \quad Update $\bm\theta^{t+1}$ via Algorithm \ref{alg:ADMM}.\\
    \quad $t\leftarrow t+1.$
    \State \textbf{Until} convergence.
    \State \textbf{Output} $\mathbf p^{\circ}=\mathbf p^{t}$, $\mathbf w_k^{\circ}=\mathbf w_k^{t}$, $\bm \theta^{\circ}=\bm\theta^{t}$.
    \end{algorithmic}
\end{algorithm}

\subsection{Alternating Optimization Framework}
We summarize the proposed alternating optimization algorithm in Algorithm \ref{alg:AO}. Specifically, the algorithm is first initialized by $\mathbf p^0$, $\mathbf w_k^0$ and $\bm\theta^0$. Then, given fixed $\mathbf p^t$, $\mathbf w_k^t$ and $\bm\theta^t$ in the $t$-th iteration, $\mathbf p^{t+1}$, $\mathbf w_k^{t+1}$ and $\bm\theta^{t+1}$ in the $(t+1)$-th iteration are updated alternatively via Algorithm \ref{alg:MMpower}, Equation (\ref{eq:optbeam}) and Algorithm \ref{alg:ADMM}, respectively. Moreover, the convergence of the AO algorithm is demonstrated in Lemma \ref{lem:AO_converge}.

\begin{Lemma}\label{lem:AO_converge}
    With the AO algorithm, the objective value of $\mathcal P1$ is non-increasing in the consecutive iterations.
\end{Lemma}
\begin{proof}
    Please refer to Appendix D.
\end{proof}

\subsection{Complexity Analysis}
The computational complexities of our proposed algorithms are discussed as follows.
\begin{itemize}
    \item For the power optimization problem, each $\mathcal P_{\mathbf p}[n+1]$ involves $K$ primal variables and $2K+1$ dual variables. Therefore, the worst-case complexity for solving $\mathcal P_{\mathbf p}[n+1]$ is $\mathcal O\left((3K+1)^{3.5}\right)$ \cite{BenTal2001cvx}. In turn, the total complexity for solving the power optimization problem is $\mathcal O\left(I_p((3K+1)^{3.5})\right)$, where $I_p$ is the number of successive convex approximation (SCA) iterations and its value is around 3 to 5 as shown in the simulation.
    \item The beamforming optimization problem is solved in closed form. Therefore, the computational complexity for beamforming optimization is negligible.
    \item Each subproblem in the ADMM algorithm for finding the feasible $\bm\theta$ is solved in closed form. Since we need to update $2K+1$ $M$-dimensional variables $\mathbf q_k$'s and $\mathbf u_k$'s, as well as $\bm\theta$ (all with closed form), the complexity of the ADMM algorithm can be regarded as $\mathcal O(I_AKM)$, where $I_A$ is the number of ADMM iterations. Moreover, the complexity of the bisection search outside the ADMM feasibility problem is $\mathcal O\left(\log_2(1/\epsilon)\right)$, where $\epsilon$ is the solution accuracy of the bisection search algorithm. Hence, the total complexity of solving the phase-shift optimization problem is $\mathcal O\left(I_{A}KM\log_2(1/\epsilon)\right)$.
    \end{itemize}

    It should be emphasized that the phase-shift optimization part in our proposed scheme has significantly lower complexity compared to the celebrated semidefinite relaxation (SDR) approach as in \cite{QWu2019IRSBeamforming}, which requires a complexity of  $\mathcal O\left(M^{4.5}\log_2(1/\epsilon)\right)$ \cite{Luo2010SDR} where $\epsilon$ is the solution accuracy. Besides, the solution achieved by SDR generally does not conform to the rank-1 constraint, and a large number of Gaussian randomizations are required to find a rank-1 solution.

    Based on the above analysis, it can be seen that the complexity of the proposed scheme is low. Moreover, hardware acceleration via FPGA or GPU can be used to further speed up the algorithm. More importantly, the power optimization problem can be considered as a nonlinear programming resource allocation problem. Such type of problem can be reformulated as a regression problem. Then a multi-task learning based feedforward neural network (MTFNN) model can be designed and trained to optimize the resource allocation problem. It has been shown in \cite{Yang2020MEC-MultiTaskLearning} that the MTFNN significantly saves the computational complexity and can be executed in real-time. The MTFNN-based approach to solving our problem is left for future investigation.

\section{Simulation Results}
In this section, we evaluate the performance of our proposed algorithms via simulations. We consider 4 users each with a learning task. The 4 learning tasks considered herein are SVM, CNN with MNIST dataset, CNN with Fashion-MNIST dataset and PointNet. The number of BS antennas varies from 10 to 50, and the number of reflecting elements of the RIS is set to 50. The total transmission time $T=10$ s, bandwidth $B=5$ MHz, total transmit power $P=1$ W, and noise power $\sigma^2=-77$ dBm. All the channels involved are assumed to be Rayleigh fading, and the channel coefficients (i.e., the elements in $\mathbf G$, $\mathbf h_{\text{d},k}$, and $\mathbf h_{\text{r},k}$, for all $k$) are normalized with zero mean and unit variance \cite{Guo2020sumrateRIS}. The pathloss exponent of the direct link, i.e., from BS to the users is 4 and the pathloss exponents of BS-RIS link and RIS-user link are set to 2.2.

\begin{figure}
    \begin{subfigure}{0.5\textwidth}
      \centering
      \includegraphics[width=\linewidth]{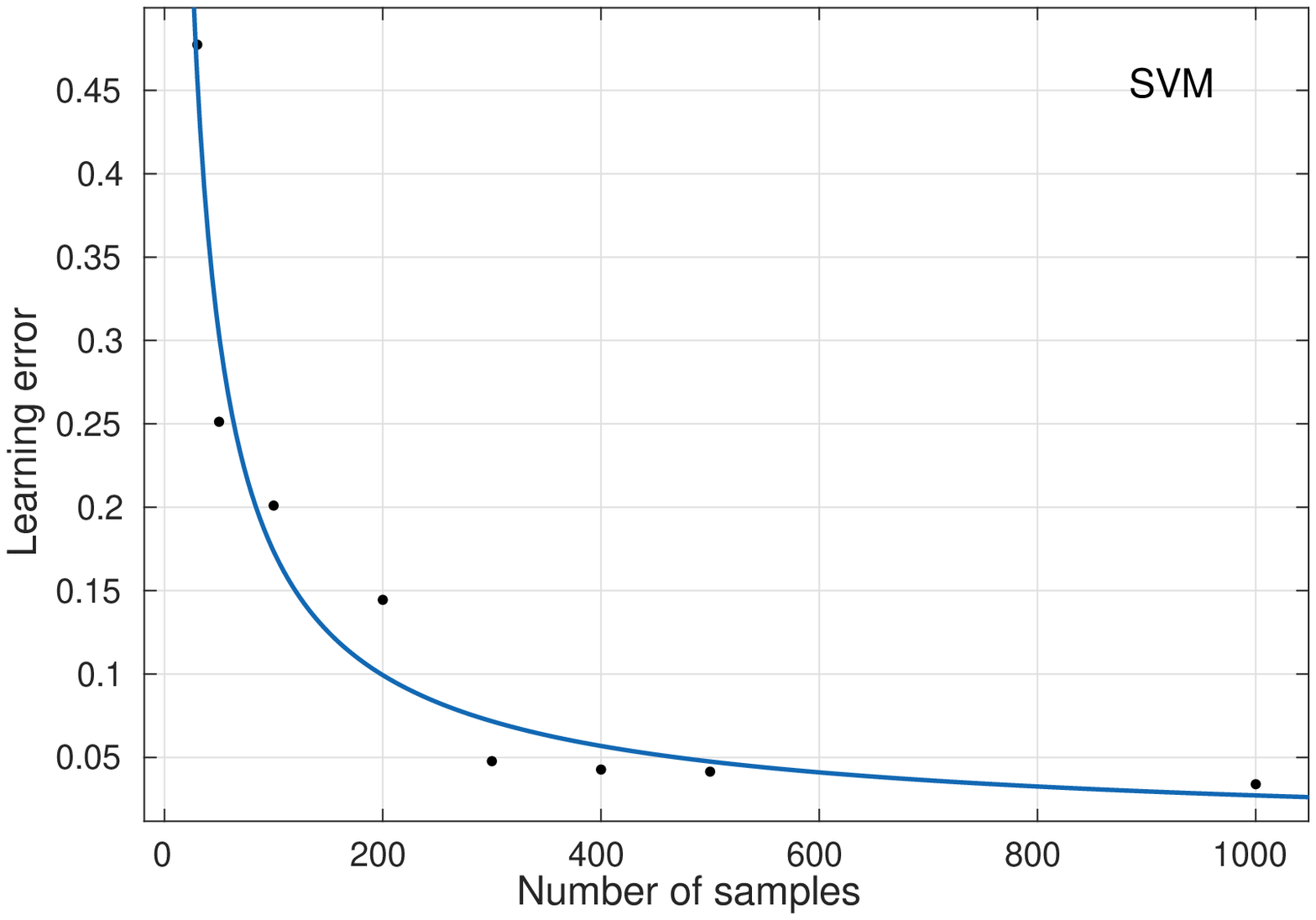}
      \caption{SVM}
      \label{fig:svm_fit}
    \end{subfigure}
    \begin{subfigure}{0.5\textwidth}
      \centering
      \includegraphics[width=\linewidth]{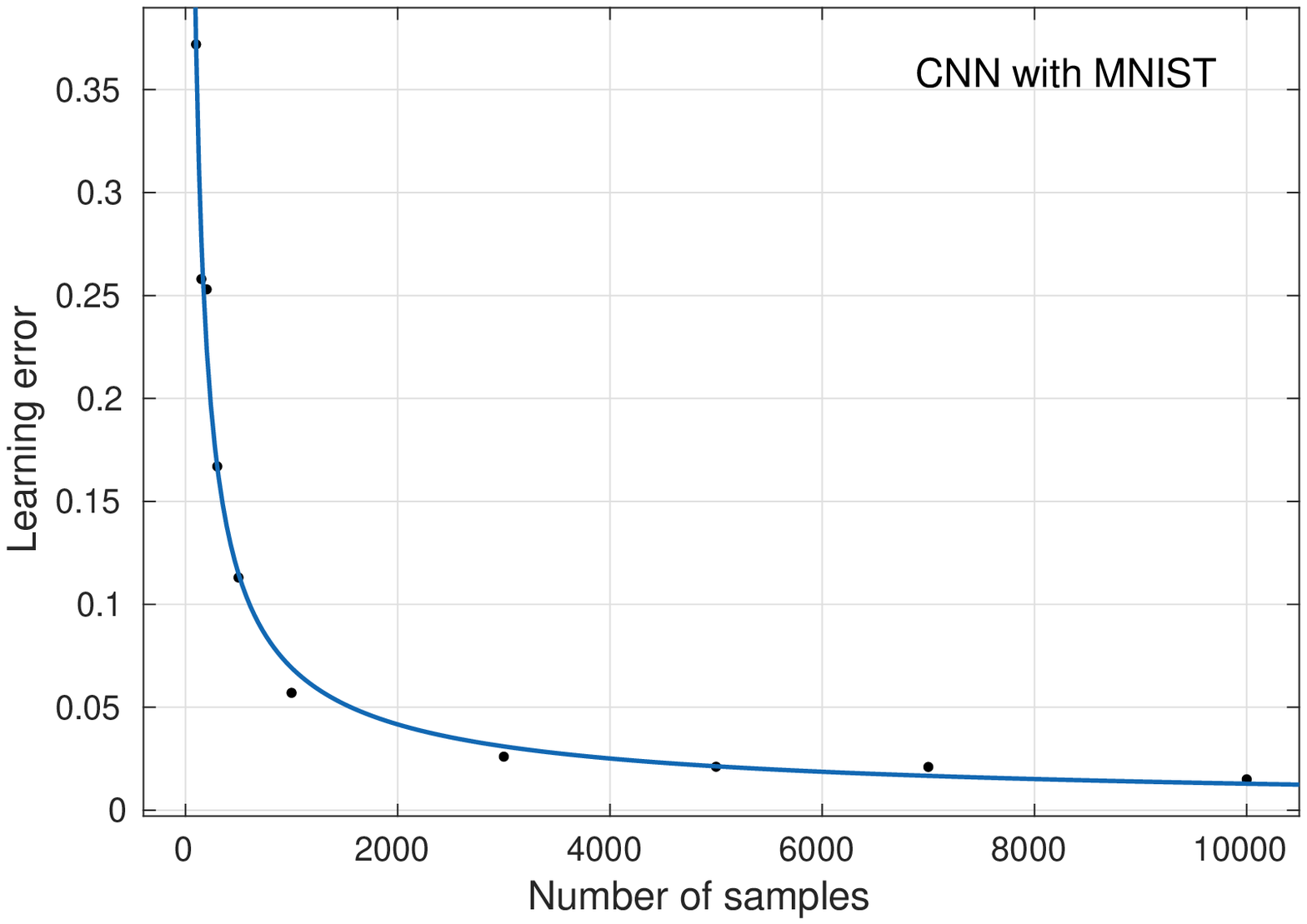}
      \caption{CNN with MNIST}
      \label{fig:mnist_fit}
    \end{subfigure}
    \newline
    \begin{subfigure}{0.5\textwidth}
      \centering
      \includegraphics[width=\linewidth]{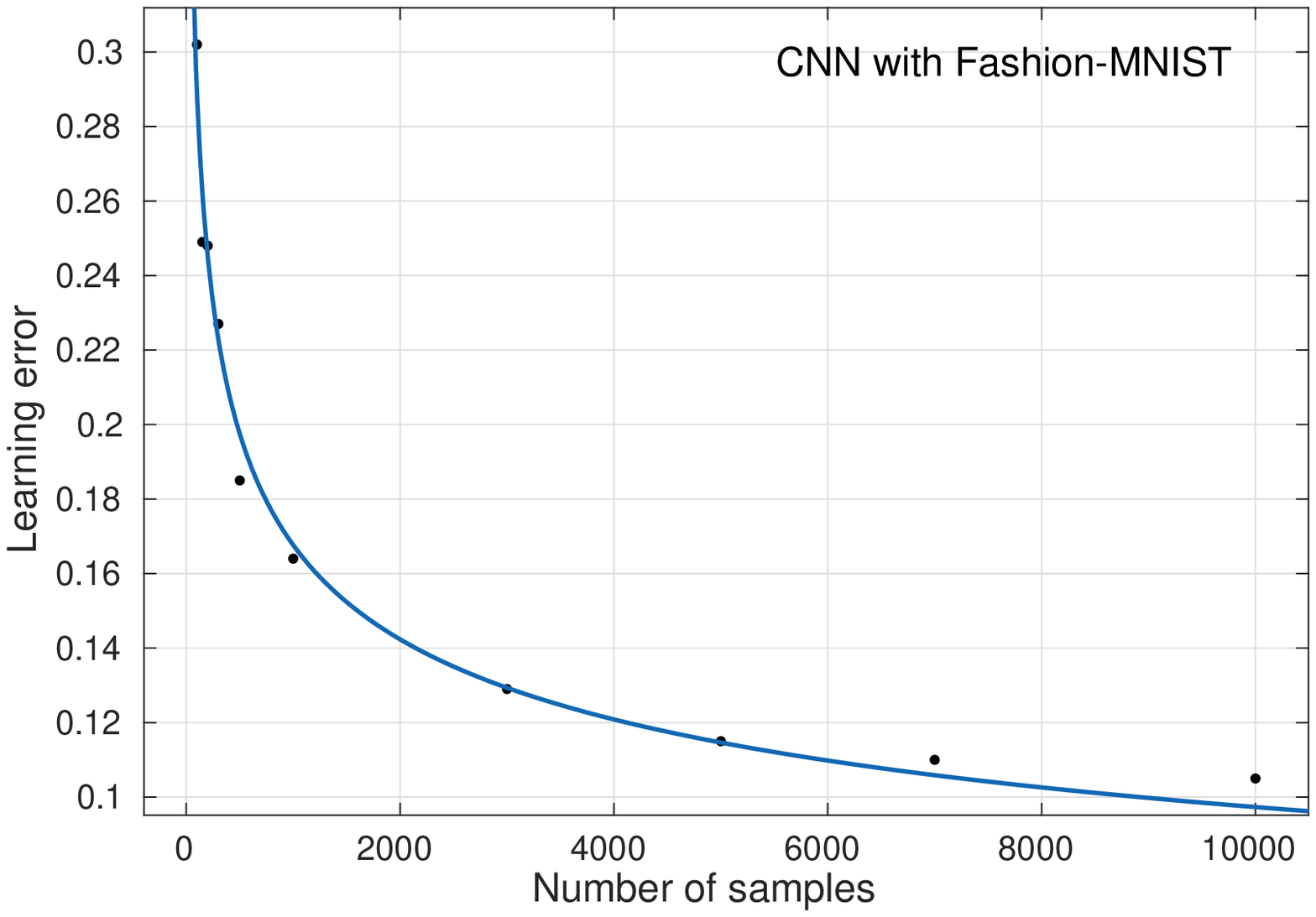}
      \caption{CNN with Fashion-MNIST}
      \label{fig:fashion_fit}
    \end{subfigure}
    \begin{subfigure}{0.5\textwidth}
      \centering
      \includegraphics[width=\linewidth]{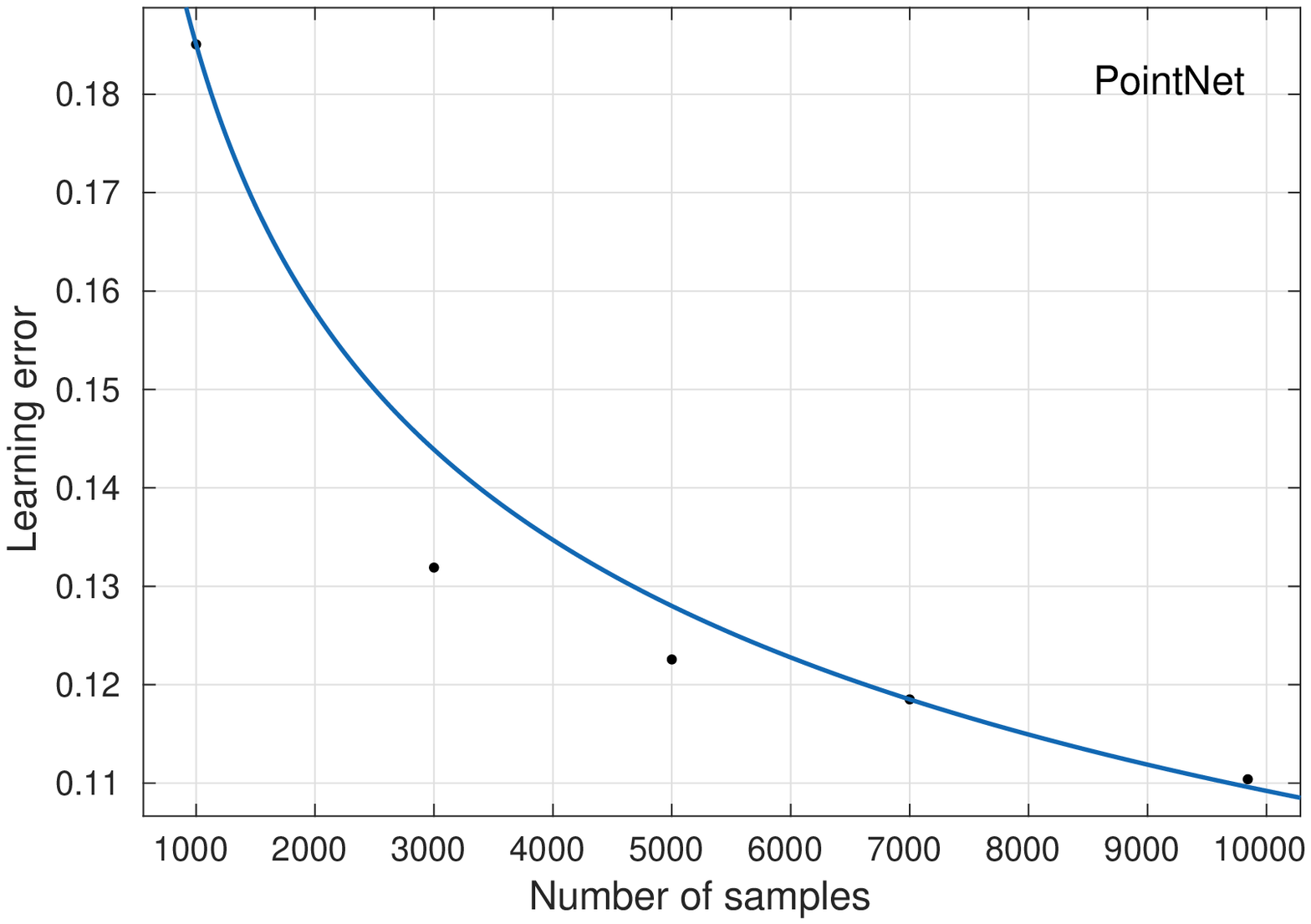}
      \caption{PointNet}
      \label{fig:point_fit}
    \end{subfigure}
    \caption{Fitting curves for the learning tasks.}
    \label{fig:fitting_curve}
\end{figure}

\begin{table}[h]
    \caption{Parameters of the nonlinear learning models for different learning tasks.}
    \centering
    \begin{tabular}{|p{2cm}|p{2cm}|p{2.5cm}|p{3.5cm}|p{2cm}|}
    \hline
    Learning Tasks & SVM & CNN with MNIST & CNN with Fashion-MNIST & PointNet\\
    \hline
    $c_k$ & 7.07 & 10.79 & 0.82& 0.96\\
    \hline
    $d_k$ & 0.81 & 0.73 & 0.23& 0.24\\
    \hline
    \end{tabular}
    \label{tab:nonlinearmodel}
\end{table}

\subsection{Parameter Fitting for the Learning Tasks}
In this part, the parameters $c_k$'s and $d_k$'s  in the nonlinear learning error models for the $K$ learning tasks are acquired by least mean square (LMS) fitting. Specifically, the SVM classifier is trained on the digits dataset in the Python Scikit-learn ML toolbox. The dataset contains 1797 images of size $8\times 8$ from 10 classes, with 5 bits (representing integers $0\sim16$) for each pixel \cite{scikit-learn}. Thus, each images needs $D_k=8\times 8\times 5+4=324$ bits. We train the SVM classifier using the first 1000 image samples with sizes $30,50,100,200,300,500,1000$, and use the last 797 image samples for testing. We record the corresponding test errors with different training sample sizes. After that, LMS fitting is applied to obtain $(c_k,d_k)$ for the SVM classifier. Then, we consider a 6-layer CNN with MNIST \cite{MNIST} and Fashion-MNIST \cite{Fashion-MNIST} datasets, respectively. The CNN consists of a $5\times 5$ convolution layer (with ReLu activation, 32 channels), a $2\times 2$ max pooling layer, another $5 \times 5$ convolution layer (with ReLu activation, 64 channels), a $2\times 2$ max pooling layer, a fully connected layer with 128 units (with ReLu activation), and a final softmax output layer (with 10 outputs). For the MNIST dataset, it consists of 70000 grayscale images (a training set of 60000 examples and a test set of 10000 examples) of handwritten digits, each with $28\times 28$ pixels. Thus, each image needs $D_k=28\times 28\times 8+4=6276$ bits. The Fashion-MNIST dataset is very similar to the MNIST dataset, except that it contains images of fashion items, such as ``T-shirt'', ``Trouser'', ``Bag'', etc., instead of handwritten digits. Each image sample of Fashion-MNIST dataset also needs $D_k=6276$ bits. We train the CNN classifier with sample sizes $100,150,200,300,500,1000,	3000,5000,7000,10000$ for both MNIST and Fashion-MNIST datasets, and record the test errors corresponding to the different training sample sizes. Then, similar LMS fitting is exploited to obtain $(c_k,d_k)$ for these two learning tasks. We also consider PointNet \cite{PointNet} as another learning task in the simulation, which applies feature transformations and aggregates point features by max pooling to classify 3D point clouds dataset ModelNet40, which contains 12311 CAD models from 40 object categories and  splits into 9843 for training and 2468 for testing. Each data sample has 2000 points with three single-precision floating-point coordinates (4 Bytes). Thus, the data size per sample is $D_k=(2000\times3\times 4+1)\times 8 = 192008$ bits. Similarly, we train the PointNet with sample sizes $1000,3000,5000,7000,9843$, and fit the result to the nonlinear learning error model to obtain $(c_k,d_k)$ for PointNet. The fitting curves are shown in Fig. \ref{fig:fitting_curve} and the parameters of the nonlinear learning error model for the above 4 learning tasks are listed in Table. \ref{tab:nonlinearmodel}.

\begin{figure}[h]
    \centering
    \begin{minipage}{.496\textwidth}
    \centering
    \includegraphics[width=\linewidth]{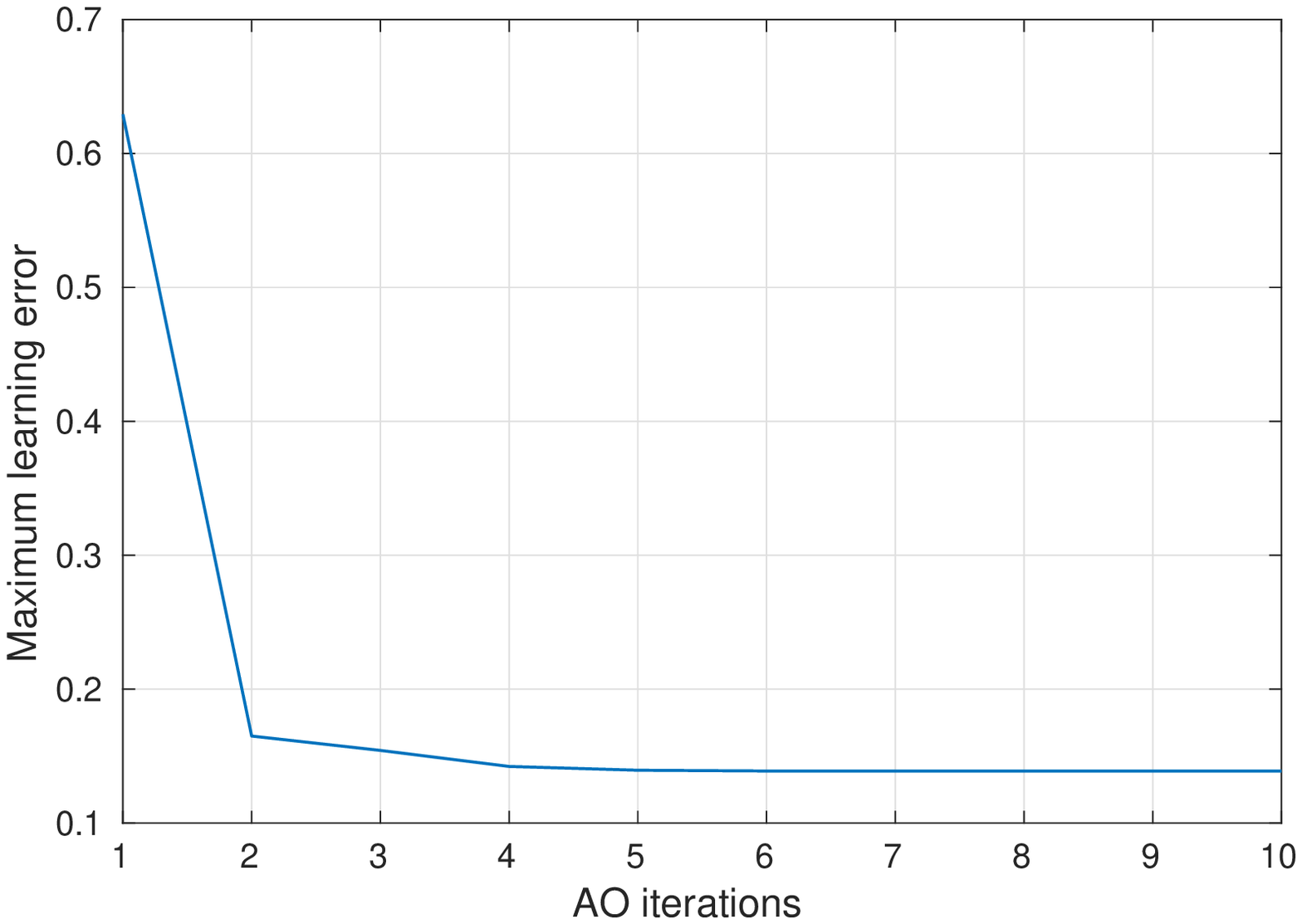}
    \caption{Convergence of the AO algorithm.}
    \label{fig:AO_converge}
    \end{minipage}
    \begin{minipage}{.496\textwidth}
    \centering
    \includegraphics[width=\linewidth]{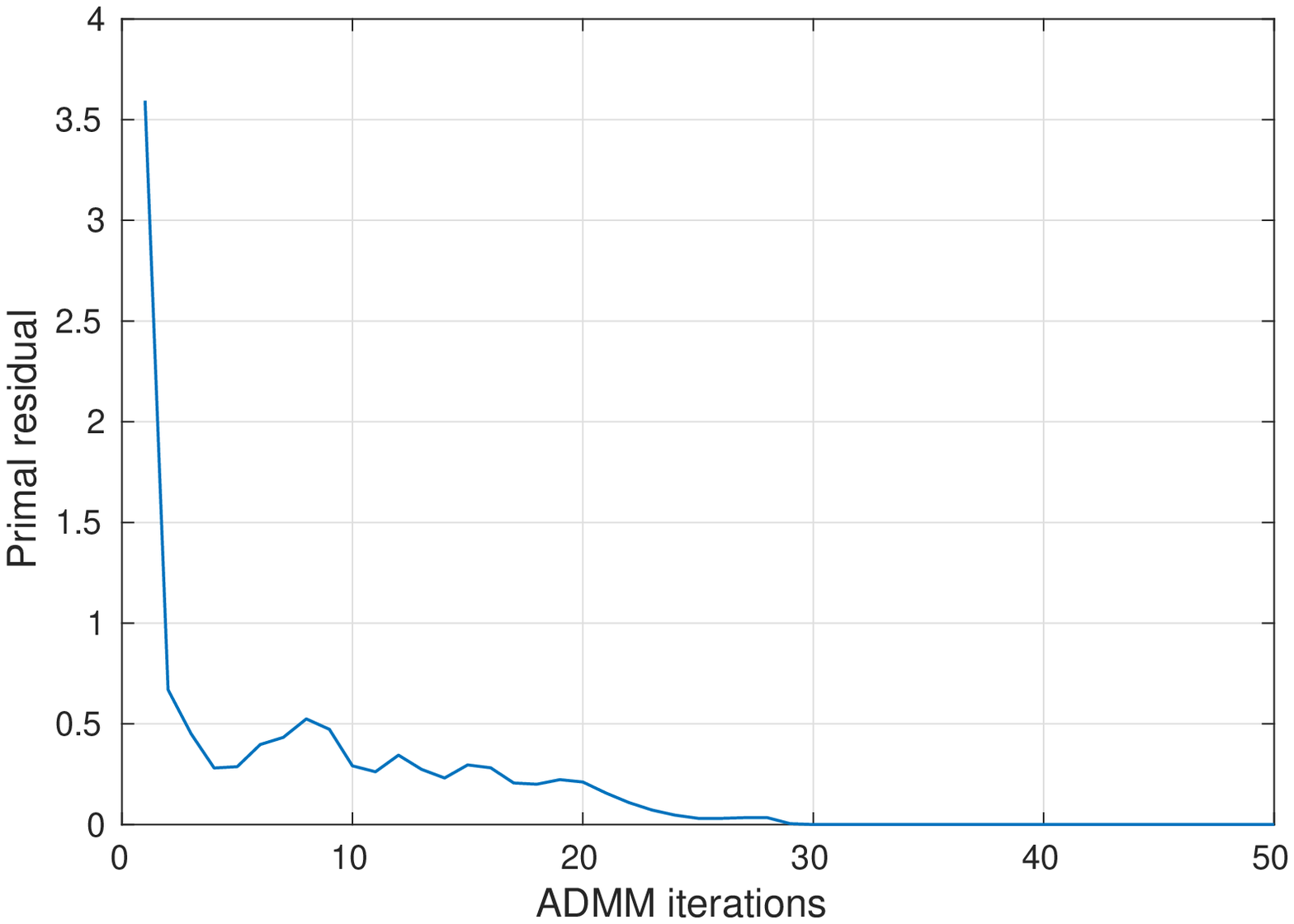}
    \caption{Convergence of ADMM.}
    \label{fig:ADMM_converge}
    \end{minipage}
\end{figure}
\subsection{Convergence of AO and ADMM algorithms}
The convergence of the AO algorithm has been proved theoretically in this paper and we further show it by simulations here. Fig. \ref{fig:AO_converge} shows that the value of objective function is non-increasing in the consecutive AO iterations, and converges after around 4 iterations, which is quite efficient.  Moreover, the convergence of the ADMM algorithm is also verified by simulations. It is shown in Fig. \ref{fig:ADMM_converge} that the primal residual concussively degrades and the ADMM algorithm converges after around 30 iterations.

\begin{figure}[h]
    \centering
    \begin{minipage}[t]{.489\textwidth}
        \centering
        \includegraphics[width=\linewidth]{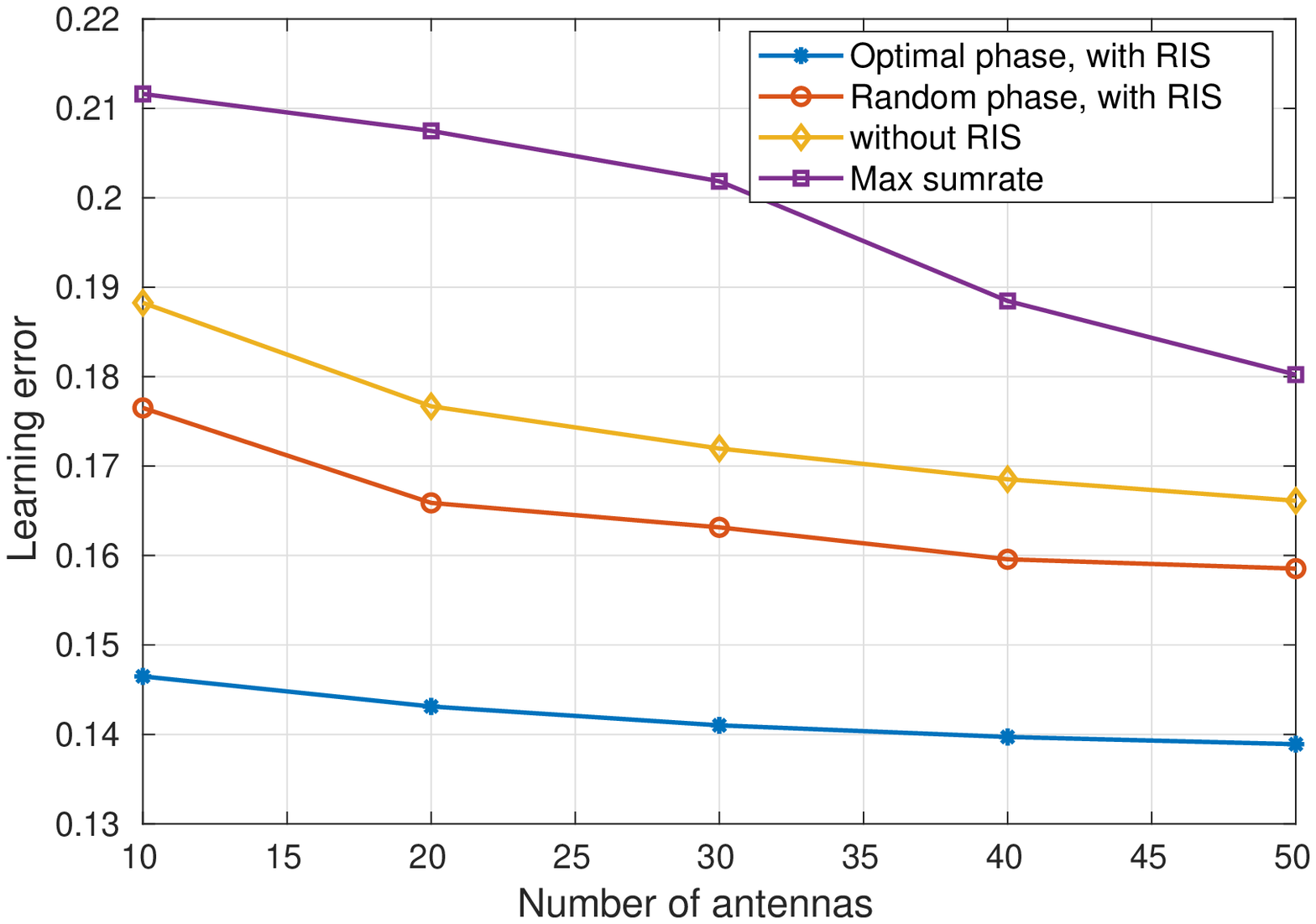}
        \caption{Learning error comparison of various benchmarks.}
        \label{fig:benchmarks}
    \end{minipage}
    \begin{minipage}[t]{.502\textwidth}
        \centering
        \includegraphics[width=\linewidth]{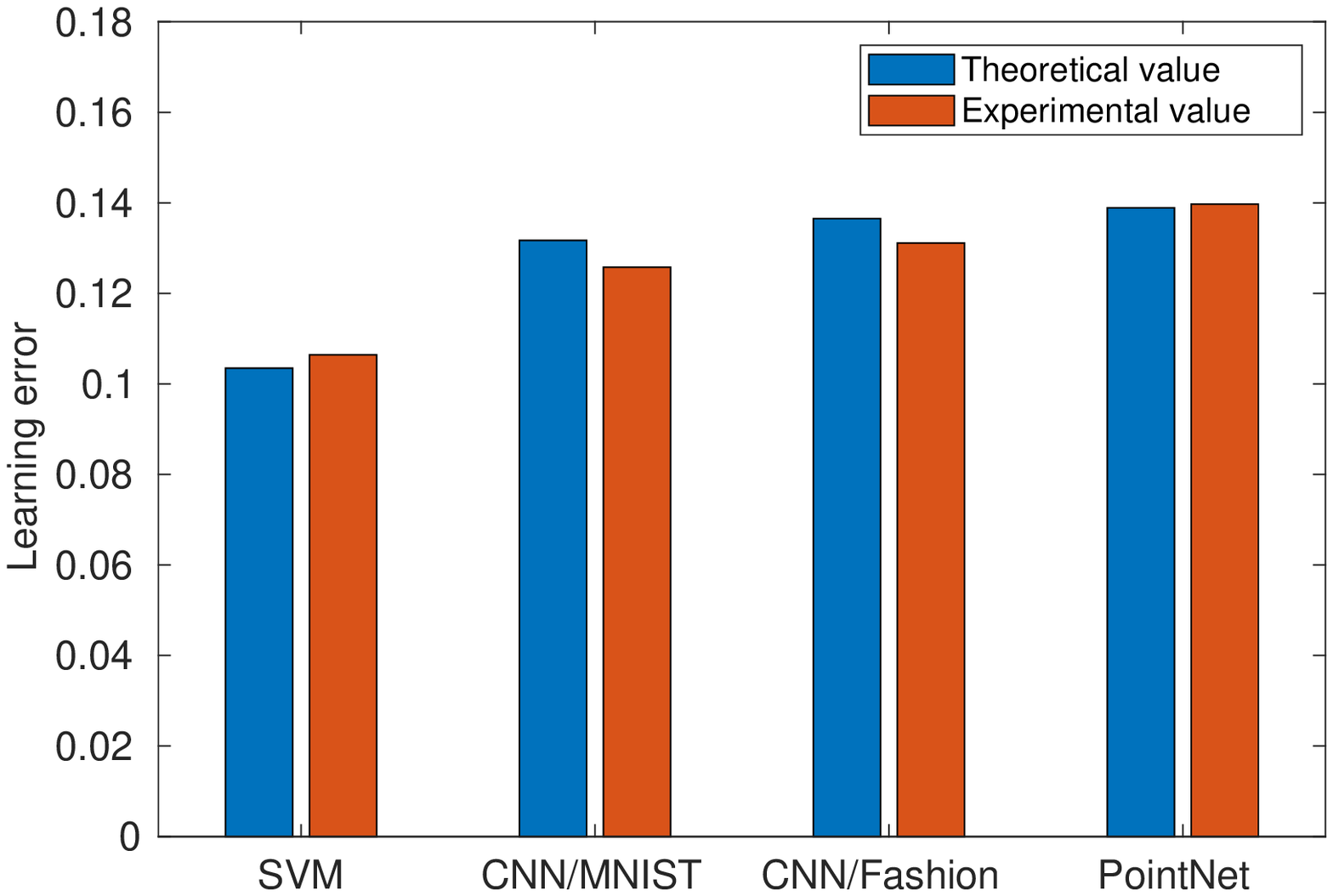}
        \caption{Theoretical learning errors v.s. Experimental learning errors. }
        \label{fig:experiment}
    \end{minipage}
\end{figure}
\subsection{Comparison with Various Benchmarks}
We demonstrate the superiority of our RIS-assisted learning-centric scheme with various benchmarks in Fig. \ref{fig:benchmarks}. The three benchmarks considered in this paper are: 1) without deploying the RIS, 2) deploying the RIS with random phase-shift matrix, and 3) maximizing the sumrate of as in conventional communication systems. It is shown that the performances of learning-centric schemes are always dramatically better than that of conventional sumrate-maximization scheme, even without the help of the RIS, which demonstrates the necessity of redesign of the resource allocation scheme in learning-driven scenarios. Fig. \ref{fig:benchmarks} also shows that with the presence of RIS, the learning performance can be improved remarkably, justifying the gain of deploying the RIS. Moreover, it can be seen that our proposed phase-shift optimization can further improve the learning accuracy significantly, validating the effectiveness of our proposed optimization algorithms.

To demonstrate the validity of the nonlinear learning error model, we compare the learning errors obtained from the theoretical error model with those obtained from real experiments. Specifically, we record the optimal number of data samples for each ML task and the corresponding theoretical learning error. Then, we use the optimized sample sizes to train the corresponding learning models, and average the resulting learning errors from 10 runs to obtain the experimental learning errors. Fig. \ref{fig:experiment} shows that the theoretical results conform to the experimental results very well.

\subsection{Experimental Results with CARLA and SECOND}
In Section V.A--C, the performance of our proposed scheme is verified based on simple ML models and datasets. However, in real-world applications, the ML task is much more complicated. For example, the 3D object detection task in autonomous driving needs to perform segmentation, classification, and box regression simultaneously. To this end, we consider objection detection in two different traffic scenarios: 1) crossroad scenario (shown in Fig. \ref{fig:car1}), and 2) downtown scenario (shown in in Fig. \ref{fig:car2}). In each scenario, an autonomous driving car senses the environment, and generates on-board video streams and LiDAR point clouds. The sensing data is then uploaded to a network edge for model training. It is assumed that the V2N data transmission is completed at the traffic lights, where the channels are stable and the road side units are equipped with RISs. However, continuous V2N data transmission is equally valid. Notice that due to different traffics and environments, the learning error models for the two scenarios are different.

\begin{figure}[h]
    \centering
    \begin{minipage}{0.49\textwidth}
        \centering
    \includegraphics[width=0.9\linewidth]{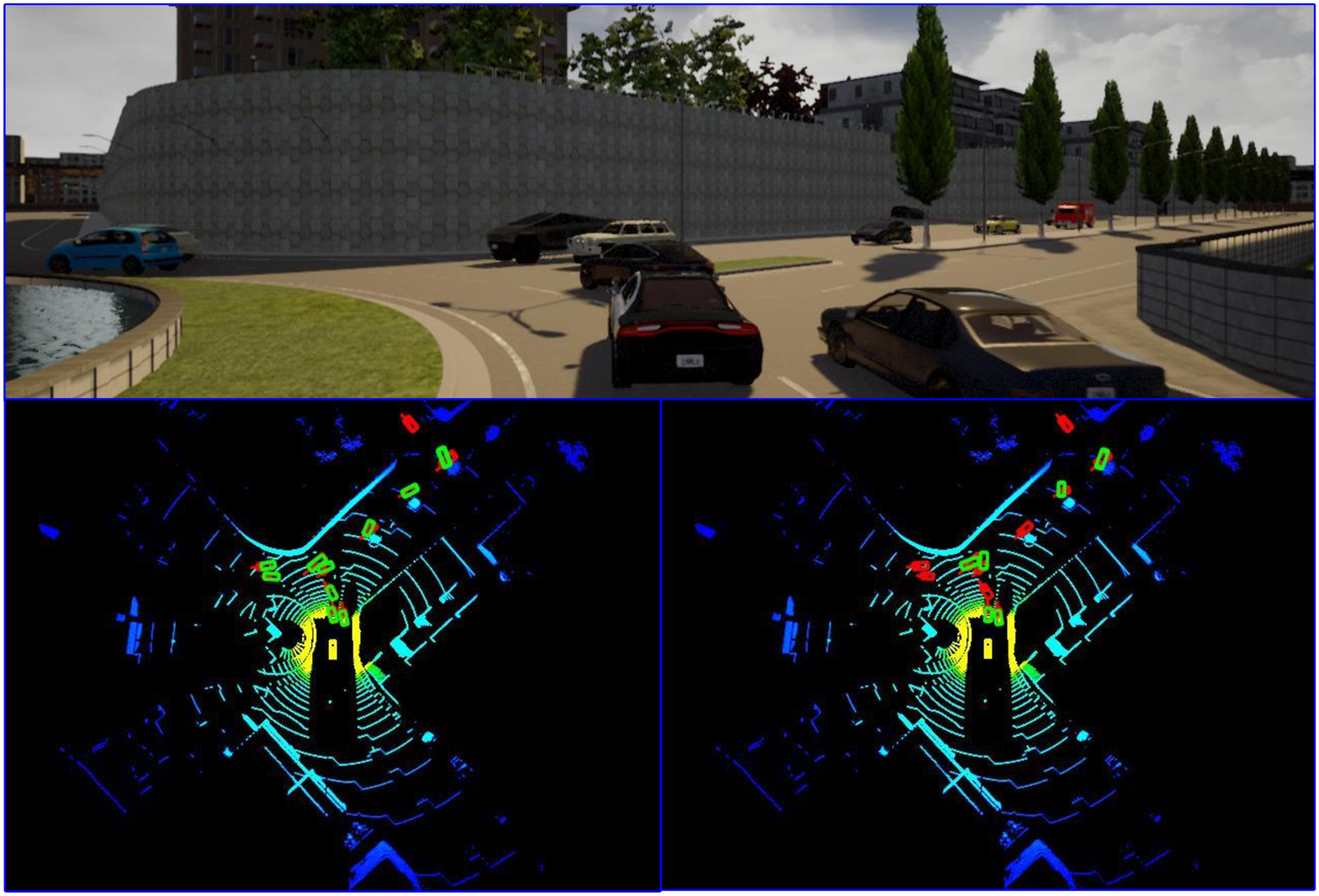}
    \caption{3D object recognition visualization for task 1\\
    Top: Crossroad scenario\\
    Bottom left: 3D object detection result with RIS\\
    Bottom right: 3D object detection result without RIS}
    \label{fig:car1}
    \end{minipage}
  \begin{minipage}{0.49\textwidth}
    \centering
    \includegraphics[width=0.9\linewidth]{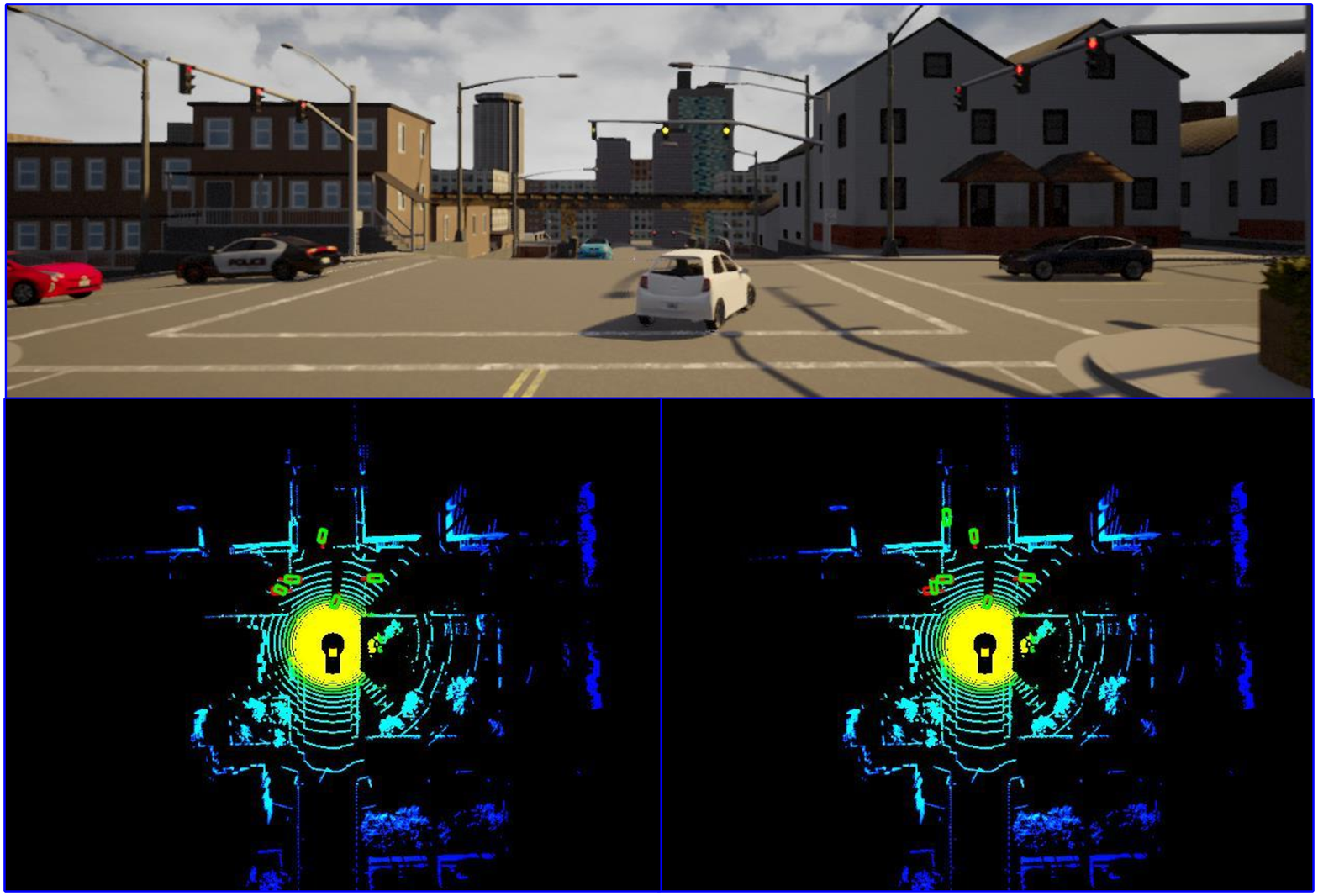}
    \caption{3D object recognition visualization for task 2\\
    Top: Downtown scenario\\
    Bottom left: 3D object detection result with RIS\\
    Bottom right: 3D object detection result without RIS}
    \label{fig:car2}
  \end{minipage}
\end{figure}
To simulate the two scenarios mentioned above, we employ the CARLA platform for dataset generation and the SECOND neural network for object detection. Specifically, CARLA is an open-source simulator that supports development, training, and validation of autonomous driving systems \cite{CARLA}. It is implemented based on Unreal Engine 4 (UE4), which provides state-of-the-art rendering quality, realistic physics, basic NPC logic, and an ecosystem of interoperable plugins. On the other hand, SECOND net is a voxel-based neural network that converts a point cloud to voxel features, and sequentially feeds the voxels into two feature encoding layers, one linear layer, one sparse CNN and one RPN (as detailed in Figure 1 in \cite{Yan2018SECOND}).
However, CARLA and SECOND are not compatible. To address this challenge, we develop a data transformation module using Python, such that the transformed dataset satisfies the KITTI standard \cite{Geiger2013IJRR}. In particular, both the crossroad and downtown datasets consist of four parts: 1) color and grayscale images captured by the high-resolution camera, 2) geographic coordinates including altitude, global orientation, velocities, accelerations, angular rates, accuracies and satellite information, 3) point cloud data from the LiDAR, and 4) object labels in the form of 3D tracklets.

Based on the above two platforms, we train the SECOND network for 50 epochs with a learning rate ranging from $10^{-4}$ to $10^{-6}$ under different number of samples. Following similar procedures in Section V.A,
the parameters in the error models for task 1 and task 2 are given by $(2.671, 0.664)$ and $(3.961, 0.501)$, respectively. By executing our proposed algorithms, the numbers of collected data samples at the edge are obtained as $20$ for task 1 and $123$ for task 2 (Note that each sample contains multiple objects as shown in Fig. \ref{fig:car1} and Fig. \ref{fig:car2}). The SECOND network is then trained with the crossroad and downtown datasets at  sample size of $20$ and $123$, respectively. The learning accuracies on the corresponding test datasets, each with $193$ unseen samples ($>1000$ objects) are $62.13\%$ for task 1 and $68.16\%$ for task 2. In contrast, the learning accuracies without the RIS are  $20.10\%$ for task 1 and $48.08\%$ for task 2. This indicates that the performance gain brought by RIS is at least $41.7\%$.

Fig. \ref{fig:car1} and Fig. \ref{fig:car2} give a visualized comparison of the learning performances for the two tasks, where the top of Figs. \ref{fig:car1} and \ref{fig:car2} are the real images captured by the two cars, and the bottom of Figs. \ref{fig:car1} and \ref{fig:car2} are the 3D object recognition results. The red blocks in the figures represent the ground-truth locations and directions of the cars on the road, and the green blocks represent the recognized locations and directions by the neural network. It can be seen that the recognition matches the ground-truth significantly better with the RIS than that without the RIS, which further demonstrates the benefit of deploying the RIS.

\section{Conclusions}
This paper investigated the RIS-assisted mobile edge computing systems with learning-driven tasks. The design of a learning-efficient system was achieved by jointly optimizing the transmit power of mobile users, the beamforming vectors of the BS and the phase-shift matrix of the RIS in an alternating optimization framework. Efficient algorithms were elaborated to address the highly nonconvex optimization problem induced by the nonlinear learning error model and unit-modulus constraints of RIS elements. Simulation and experimental results demonstrated the validity of the learning error model and superiority of our proposed scheme over various benchmarks. A unified communication-training-inference platform is developed based on the CARLA platform and the SECOND network, and a use case (3D object detection in autonomous driving) for the proposed scheme is demonstrated on the developed platform.

\appendix
\subsection{Proof of Proposition \ref{prop:power}}
1) \textbf{Upper bound condition:}
Applying first-order Taylor expansion to $\ln x$ at $x_0$, we have $\ln x\leq \ln x_0+\frac{x}{x_0}-1$. Thus, we have
\begin{align}\label{eq:power_taylor}
    \ln\left(\sum_{i=1,1\neq k}^K|\mathbf w_k\hermconj\mathbf h_i(\mathbf \Theta)|^2p_i+\sigma^2\right)\leq &\ln\left(\sum_{i=1,1\neq k}^K|\mathbf w_k\hermconj\mathbf h_i(\mathbf \Theta)|^2p_i^{\star}+\sigma^2\right)\nonumber\\
    &+\frac{\sum_{i=1,1\neq k}^K|\mathbf w_k\hermconj\mathbf h_i(\mathbf \Theta)|^2p_i+\sigma^2}{\sum_{i=1,1\neq k}^K|\mathbf w_k\hermconj\mathbf h_i(\mathbf \Theta)|^2p_i^{\star}+\sigma^2}-1,
\end{align}
where $\mathbf h_i(\mathbf \Theta)=\mathbf h_{\text{d},i}+\mathbf G\hermconj \mathbf{\Theta}\hermconj \mathbf h_{\text{r},i}$, for $i=1,\cdots,K$.

Multiplying $-1$ and adding $\ln\left(\sum_{i=1}^K|\mathbf w_k\hermconj\mathbf h_i(\mathbf \Theta)|^2p_i+\sigma^2\right)$ on both sides of (\ref{eq:power_taylor}), we obtain
\begin{align}\label{eq:power_ineq}
    \ln\left(1+\frac{|\mathbf w_k\hermconj\mathbf h_k(\mathbf \Theta)|^2p_k}{\sum_{i=1,1\neq k}^K|\mathbf w_k\hermconj\mathbf h_i(\mathbf \Theta)|^2p_i+\sigma^2}\right)\geq &\ln\left(\sum_{i=1}^K|\mathbf w_k\hermconj\mathbf h_i(\mathbf \Theta)|^2p_i+\sigma^2\right)\nonumber\\
    &-\frac{\sum_{i=1,1\neq k}^K|\mathbf w_k\hermconj\mathbf h_i(\mathbf \Theta)|^2p_i+\sigma^2}{\sum_{i=1,1\neq k}^K|\mathbf w_k\hermconj\mathbf h_i(\mathbf \Theta)|^2p_i^{\star}+\sigma^2}\nonumber\\
    &-\ln\left(\sum_{i=1,1\neq k}^K|\mathbf w_k\hermconj\mathbf h_i(\mathbf \Theta)|^2p_i^{\star}+\sigma^2\right)+1.
\end{align}
Substituting the result of (\ref{eq:power_ineq}) into $\widetilde\Phi_m(\mathbf p|\mathbf p^{\star})$ in (\ref{eq:power_bound}), and since $c_mx^{-d_m}$ is a decreasing function of $x$, we immediately have
\begin{equation}
    \widetilde\Phi_m(\mathbf p|\mathbf p^{\star})\geq c_m\left[\sum_{k=1}^K\frac{\eta_{k,m}BT}{D_m}\log_2\left(1+\frac{|\mathbf w_k\hermconj \mathbf h_k(\mathbf{\Theta})|^2p_k}{\sum_{i=1, i\neq k}^K |\mathbf w_k\hermconj \mathbf h_{i}(\mathbf{\Theta})|^2 p_{i}+\sigma^2}\right)\right]^{-d_m}=\Phi_m(\mathbf p)
\end{equation}

2) \textbf{Convexity:}
Notice that we can regard $\widetilde{\Phi}_m(\mathbf p|\mathbf p^{\star})$ as a composition function of $h$ and $g$, i.e., $\widetilde{\Phi}_m(\mathbf p|\mathbf p^{\star})=h(g(\mathbf p|\mathbf p^{\star}))$, where $h(x)=c_mx^{-d_m}$ and
\begin{align}
    g(\mathbf p|\mathbf p^{\star})=&\sum_{k=1}^K\frac{\eta_{k,m}BT}{D_m \ln2}\bigg[\ln\left(\sum_{i=1}^K|\mathbf w_k\hermconj \mathbf h_{i}(\mathbf{\Theta})|^2 p_{i}+\sigma^2\right)-\frac{\sum_{i=1,i\neq k}^K|\mathbf w_k\hermconj \mathbf h_{i}(\mathbf{\Theta})|^2 p_{i}+\sigma^2}{\sum_{i=1,i\neq k}^K|\mathbf w_k\hermconj \mathbf h_{i}(\mathbf{\Theta})|^2 p_{i}^{\star}+\sigma^2}\nonumber\\
    &-\ln\bigg(\sum_{i=1,i\neq k}^K|\mathbf w_k\hermconj \mathbf h_{i}(\mathbf{\Theta})|^2 p_{i}^{\star}+\sigma^2\bigg)+1\bigg].
\end{align}
It can be easily observed that $h(x)$ is convex and nonincreasing. Adding to the fact that $g(\mathbf p|\mathbf p^{\star})$ is concave in $\mathbf p$, the convexity of $\widetilde{\Phi}_m(\mathbf p|\mathbf p^{\star})$ is immediately established via composition rule \cite{boyd2004convex}.\\

3) \textbf{Local condition:} The proof of $\widetilde{\Phi}_m(\mathbf p^{\star}|\mathbf p^{\star})=\Phi_m(\mathbf p^{\star})$ is straightforward by putting $\mathbf p=\mathbf p^{\star}$ into the definition of $\widetilde{\Phi}_m(\mathbf p|\mathbf p^{\star})$ and $\Phi_m(\mathbf p)$.
On the other hand, in order to prove $\nabla_{\mathbf p}\widetilde{\Phi}_m(\mathbf p^{\star}|\mathbf p^{\star})=\nabla_{\mathbf p}\Phi_m(\mathbf p^{\star})$, we can calculate their derivatives which are given by
\begin{align}
    \label{eq:derivative_p_bound}
    \nabla_{p_j}\widetilde\Phi_m(\mathbf p|\mathbf p^{\star})=&-c_m d_m\bigg\{\sum_{k=1}^K\frac{\eta_{k,m}BT}{D_m\ln2}\bigg[\ln\left(\sum_{i=1}^K|\mathbf w_k\hermconj \mathbf h_i(\mathbf \Theta)|^2p_i+\sigma^2\right)\nonumber\\
    &-\frac{\sum_{i=1,i\neq k}^K|\mathbf w_k\hermconj \mathbf h_{i}(\mathbf{\Theta})|^2 p_{i}+\sigma^2}{\sum_{i=1,i\neq k}^K|\mathbf w_k\hermconj \mathbf h_{i}(\mathbf{\Theta})|^2 p_{i}^{\star}+\sigma^2}-\ln\bigg(\sum_{i=1,i\neq k}^K|\mathbf w_k\hermconj \mathbf h_{i}(\mathbf{\Theta})|^2 p_{i}^{\star}+\sigma^2\bigg)+1\bigg]\bigg\}^{-d_m-1}\nonumber\\
    &\times \sum_{k=1}^K\frac{\eta_{k,m}BT}{D_m\ln2}\bigg[\frac{|\mathbf w_k\hermconj\mathbf h_j(\mathbf \Theta)|^2}{\sum_{i=1}^K|\mathbf w_k\hermconj\mathbf h_i(\mathbf \theta)|^2p_i+\sigma^2}-\frac{|\mathbf w_k\hermconj \mathbf h_j(\mathbf\Theta)|^2\mathbb I(j\neq k)}{\sum_{i=1,i\neq k}|\mathbf w_k\hermconj\mathbf h_i(\mathbf\Theta)|^2p_i^{\star}+\sigma^2}\bigg]
\end{align}
\begin{align}
    \label{eq:derivative_p}
    \nabla_{p_j}\Phi(\mathbf p)=&-c_md_m\bigg[\sum_{k=1}^K\frac{\eta_{k,m}BT}{D_m\ln2}\ln\left(1+\frac{|\mathbf w_k\hermconj \mathbf h_k(\mathbf{\Theta})|^2p_k}{\sum_{i=1, i\neq k}^K |\mathbf w_k\hermconj \mathbf h_{i}(\mathbf{\Theta})|^2 p_{i}+\sigma^2}\right)\bigg]^{-d_m-1}\nonumber\\
    &\times \sum_{k=1}^K\frac{\eta_{k,m}BT}{D_m\ln2}\bigg[\frac{|\mathbf w_k\hermconj\mathbf h_j(\mathbf \Theta)|^2}{\sum_{i=1}^K|\mathbf w_k\hermconj\mathbf h_i(\mathbf \theta)|^2p_i+\sigma^2}-\frac{|\mathbf w_k\hermconj \mathbf h_j(\mathbf\Theta)|^2\mathbb I(j\neq k)}{\sum_{i=1,i\neq k}|\mathbf w_k\hermconj\mathbf h_i(\mathbf\Theta)|^2p_i+\sigma^2}\bigg]
\end{align}
Substituting $\mathbf p=\mathbf p^{\star}$ into (\ref{eq:derivative_p}), we immediately have $\nabla_{\mathbf p}\widetilde{\Phi}_m(\mathbf p^{\star}|\mathbf p^{\star})=\nabla_{\mathbf p}\Phi_m(\mathbf p^{\star})$.

\subsection{Proof of Lemma \ref{lem:opt_bf}}
Denoting $\mathbf h_i=\mathbf h_{\text{d},i}+\mathbf G\hermconj \mathbf{\Theta}\hermconj \mathbf h_{\text{r},i}$, problem $\mathcal P_{\mathbf w_k}$ can be written as
\begin{subequations}\label{prob:bf_equ}
    \begin{align}
        \max_{\mathbf w_k} \quad &\frac{\mathbf w_k\hermconj\mathbf h_k\mathbf h_k\hermconj \mathbf w_kp_k/\sigma^2}{\sum_{i=1,i\neq k}^K\mathbf w_k\hermconj\mathbf h_i\mathbf h_i\hermconj \mathbf w_k p_i/\sigma^2+\mathbf w_k\hermconj\mathbf w_k}\\
        \st \quad&\mathbf w_k\hermconj\mathbf w_k=1.
    \end{align}
\end{subequations}
It can be observed that the objective function of (\ref{prob:bf_equ}) remains unchanged when scaling $\mathbf w_k$ by any positive factor. Hence, we can safely remove the constraint and then scale the resultant $\mathbf w_k$ such that $\mathbf w_k\hermconj\mathbf w_k=1$ satisfies. Denote $\mathbf\Xi=\mathbf h_k\mathbf h_k\hermconj p_k/\sigma^2$, and $\mathbf\Gamma=\sum_{i=1}^K\mathbf h_i\mathbf h_i\hermconj p_i/\sigma^2+\mathbf I_N$. Then, problem (\ref{prob:bf_equ}) can be equivalently written in the following compact form.
\begin{align}\label{prob:bf_equ1}
    \max_{\mathbf w_k}\quad \frac{\mathbf w_k\hermconj\mathbf\Xi\mathbf w_k}{\mathbf w_k\hermconj\mathbf\Gamma\mathbf w_k-\mathbf w_k\hermconj\mathbf\Xi\mathbf w_k}.
\end{align}
Since the SINR must be a positive number, maximizing $\frac{\mathbf w_k\hermconj\mathbf\Xi\mathbf w_k}{\mathbf w_k\hermconj\mathbf\Gamma\mathbf w_k-\mathbf w_k\hermconj\mathbf\Xi\mathbf w_k}$ is equivalent to minimizing its reciprocal $\frac{\mathbf w_k\hermconj\mathbf\Gamma\mathbf w_k}{\mathbf w_k\hermconj\mathbf\Xi\mathbf w_k}-1$. Finally, problem (\ref{prob:bf_equ1}) is equivalent to
\begin{align}\label{prob:bf_equ2}
    \max_{\mathbf w_k} \quad \frac{\mathbf w_k\hermconj\mathbf\Xi\mathbf w_k}{\mathbf w_k\hermconj\mathbf\Gamma\mathbf w_k}.
\end{align}

In order to solve problem (\ref{prob:bf_equ2}), we introduce a new variable $\mathbf r=\mathbf\Gamma^{1/2}\mathbf w_k$. Then, problem (\ref{prob:bf_equ2}) is transformed to
\begin{align}\label{prob:bf_equ3}
    \max_{\mathbf r}\quad\frac{\mathbf r\hermconj\mathbf \Gamma^{-1/2}\mathbf\Xi\mathbf\Gamma^{-1/2}\mathbf r}{\mathbf r\hermconj\mathbf r},
\end{align}
which is a standard eigenvalue problem, and its optimal solution is the dominant eigenvector of $\mathbf \Gamma^{-1/2}\mathbf\Xi\mathbf\Gamma^{-1/2}$, i.e., $\mathbf r^*=\mathbf\Gamma^{-1/2}\mathbf h_k$. Substituting $\mathbf r=\mathbf r^*$ into $\mathbf r=\mathbf\Gamma^{1/2}\mathbf w_k$, we have $\mathbf w_k=\mathbf \Gamma^{-1}\mathbf h_k=\left(\mathbf I_N+\sum_{i=1}^K\frac{p_{i}}{\sigma^2}\mathbf h_{i}\mathbf h_{i}\hermconj\right)^{-1}\mathbf h_k$. Scaling $\mathbf w_k$ such that $\mathbf w_k\hermconj\mathbf w_k=1$, we have
\begin{align}
    \mathbf w_k^{\diamond}=\frac{\left(\mathbf I_N+\sum_{i=1}^K\frac{p_{i}}{\sigma^2}\mathbf h_{i}\mathbf h_{i}\hermconj\right)^{-1}\mathbf h_k}{\left\|\left(\mathbf I_N+\sum_{i=1}^K\frac{p_{i}}{\sigma^2}\mathbf h_{i}\mathbf h_{i}\hermconj\right)^{-1}\mathbf h_k\right\|_2}.
\end{align}
Thus, the proof is finished.

\subsection{Proof of Lemma \ref{lem:opt_q}}
Since strong duality holds for QCQP problems with one constraint as proved in \cite{boyd2004convex}, we can solve the dual problem of (\ref{prob:q_qcqp1_equ}). The Lagrangian of (\ref{prob:q_qcqp1_equ}) is
\begin{align}
    \mathcal L(\mathbf q,\mu)=\|\tilde{\mathbf q}-\bm\zeta^t\|^2+\mu(\tilde{\mathbf q}\hermconj\mathbf\Lambda\tilde{\mathbf q}-2\re\{\tilde(\mathbf b)\hermconj\tilde{\mathbf q}\}-\tau).
\end{align}
Taking derivative of $\mathcal L(\mathbf q,\mu)$ with respect to $\mathbf q$, we have
\begin{align}
    \frac{\partial \mathcal L(\mathbf q,\mu)}{\partial \mathbf q}=\tilde{\mathbf q}-\mathbf{\zeta}^t+\mu\mathbf{\Lambda}\tilde{\mathbf q}-\mu\tilde{\mathbf b}.
\end{align}
Setting $\frac{\partial \mathcal L(\mathbf q,\mu)}{\partial \mathbf q}=0$, we obtain the optimal $\tilde{\mathbf q}$ as
\begin{align}
    \tilde{\mathbf q}^*=(\mathbf I+\mu\mathbf\Lambda)^{-1}(\tilde{\bm\zeta}+\mu\tilde{\mathbf b}).
\end{align}
Substituting the above equation back to the equality constraint in (\ref{prob:q_qcqp1_equ}), it becomes a nonlinear equation with respect to $\mu$:
\begin{align}
    \chi(\mu)=\sum_{m=1}^M\lambda_m\left|\frac{\tilde{\zeta}_m+\mu\tilde b_m}{1+\mu\lambda_m}\right|^2-2\re\left\{\sum_{m=1}^M\tilde b_m^*\frac{\tilde{\zeta}_m+\mu\tilde b_m}{1+\mu\lambda_m}\right\}-\tau,
\end{align}
where $\lambda_m$ is the $m$-th diagonal entry of $\mathbf\Lambda$. Thus, the proof is finished.
\subsection{Proof of Lemma \ref{lem:AO_converge}}
For ease of notation, we denote the objective function of $\mathcal P1$ as $g(\mathbf p,\mathbf w,\bm\theta)$. Assume $\mathbf p^t$, $\mathbf w^t$ and $\bm\theta^t$ are obtained by the corresponding optimization problems in the $t$-th iteration, respectively. Then, we have
\begin{align}
g(\mathbf p^{t+1},\mathbf w^t,\bm\theta^t)=\min_{\mathbf p} g(\mathbf p,\mathbf w^t,\bm\theta^t)\leq g(\mathbf p^{t},\mathbf w^t,\bm\theta^t).
\end{align}
Analogously, it holds that
\begin{align}
    g(\mathbf p^{t+1},\mathbf w^{t+1},\bm\theta^t)=\min_{\mathbf w} g(\mathbf p^{t+1},\mathbf w,\bm\theta^t)\leq g(\mathbf p^{t+1},\mathbf w^t,\bm\theta^t)\leq g(\mathbf p^{t},\mathbf w^t,\bm\theta^t).
\end{align}
Finally, it is established that
\begin{align}
    g(\mathbf p^{t+1},\mathbf w^{t+1},\bm\theta^{t+1})\leq g(\mathbf p^t,\mathbf w^t,\bm\theta^t).
\end{align}
Therefore, the objective value of $\mathcal P1$ is non-increasing in the consecutive AO iterations, which indicates that the AO algorithm is guaranteed to converge.

\bibliographystyle{IEEEtran}
\bibliography{RISEL_TCCN.bbl}

\begin{thebibliography}{10}
\providecommand{\url}[1]{#1}
\csname url@samestyle\endcsname
\providecommand{\newblock}{\relax}
\providecommand{\bibinfo}[2]{#2}
\providecommand{\BIBentrySTDinterwordspacing}{\spaceskip=0pt\relax}
\providecommand{\BIBentryALTinterwordstretchfactor}{4}
\providecommand{\BIBentryALTinterwordspacing}{\spaceskip=\fontdimen2\font plus
\BIBentryALTinterwordstretchfactor\fontdimen3\font minus
  \fontdimen4\font\relax}
\providecommand{\BIBforeignlanguage}[2]{{%
\expandafter\ifx\csname l@#1\endcsname\relax
\typeout{** WARNING: IEEEtran.bst: No hyphenation pattern has been}%
\typeout{** loaded for the language `#1'. Using the pattern for}%
\typeout{** the default language instead.}%
\else
\language=\csname l@#1\endcsname
\fi
#2}}
\providecommand{\BIBdecl}{\relax}
\BIBdecl

\bibitem{Mao2017MECsurvey}
Y.~{Mao}, C.~{You}, J.~{Zhang}, K.~{Huang}, and K.~B. {Letaief}, ``A survey on
  mobile edge computing: The communication perspective,'' \emph{IEEE
  Communications Surveys Tutorials}, vol.~19, no.~4, pp. 2322--2358,
  Forthquater 2017.

\bibitem{Zhou2019EI}
Z.~{Zhou}, X.~{Chen}, E.~{Li}, L.~{Zeng}, K.~{Luo}, and J.~{Zhang}, ``Edge
  intelligence: Paving the last mile of artificial intelligence with edge
  computing,'' \emph{Proceedings of the IEEE}, vol. 107, no.~8, pp. 1738--1762,
  Aug. 2019.

\bibitem{Li2019EdgeAI}
E.~{Li}, L.~{Zeng}, Z.~{Zhou}, and X.~{Chen}, ``Edge {AI}: On-demand
  accelerating deep neural network inference via edge computing,'' \emph{IEEE
  Transactions on Wireless Communications}, vol.~19, no.~1, pp. 447--457, Jan.
  2020.

\bibitem{Zhu2020ELsurvey}
G.~{Zhu}, D.~{Liu}, Y.~{Du}, C.~{You}, J.~{Zhang}, and K.~{Huang}, ``Toward an
  intelligent edge: Wireless communication meets machine learning,'' \emph{IEEE
  Communications Magazine}, vol.~58, no.~1, pp. 19--25, Jan. 2020.

\bibitem{Yu2020IE}
S.~{Yu}, X.~{Chen}, L.~{Yang}, D.~{Wu}, M.~{Bennis}, and J.~{Zhang},
  ``Intelligent edge: Leveraging deep imitation learning for mobile edge
  computation offloading,'' \emph{IEEE Wireless Communications}, vol.~27,
  no.~1, pp. 92--99, Feb. 2020.

\bibitem{Chen2019DLwithEdgeComp}
J.~{Chen} and X.~{Ran}, ``Deep learning with edge computing: A review,''
  \emph{Proceedings of the IEEE}, vol. 107, no.~8, pp. 1655--1674, Aug. 2019.

\bibitem{Wang2020EdgeLearningTWC}
S.~{Wang}, Y.~{Wu}, M.~{Xia}, R.~{Wang}, and H.~V. {Poor}, ``Machine
  intelligence at the edge with learning centric power allocation,'' \emph{IEEE
  Transactions on Wireless Communications}, pp. 1--1, 2020.

\bibitem{Wang2019FL}
S.~{Wang}, T.~{Tuor}, T.~{Salonidis}, K.~K. {Leung}, C.~{Makaya}, T.~{He}, and
  K.~{Chan}, ``Adaptive federated learning in resource constrained edge
  computing systems,'' \emph{IEEE Journal on Selected Areas in Communications},
  vol.~37, no.~6, pp. 1205--1221, Jun. 2019.

\bibitem{Gunduz2019MLintheair}
D.~{Gunduz}, P.~{de Kerret}, N.~D. {Sidiropoulos}, D.~{Gesbert}, C.~R.
  {Murthy}, and M.~{van der Schaar}, ``Machine learning in the air,''
  \emph{IEEE Journal on Selected Areas in Communications}, vol.~37, no.~10, pp.
  2184--2199, Oct. 2019.

\bibitem{Zhu2019BroadbandAggr}
G.~{Zhu}, Y.~{Wang}, and K.~{Huang}, ``Broadband analog aggregation for
  low-latency federated edge learning,'' \emph{IEEE Transactions on Wireless
  Communications}, vol.~19, no.~1, pp. 491--506, Jan. 2020.

\bibitem{Sun2020FLIoT}
H.~{Sun}, S.~{Li}, F.~R. {Yu}, Q.~{Qi}, J.~{Wang}, and J.~{Liao}, ``Towards
  communication-efficient federated learning in the internet of things with
  edge computing,'' \emph{IEEE Internet of Things Journal}, pp. 1--1, 2020.

\bibitem{Tran2019FL}
N.~H. {Tran}, W.~{Bao}, A.~{Zomaya}, M.~N.~H. {Nguyen}, and C.~S. {Hong},
  ``Federated learning over wireless networks: Optimization model design and
  analysis,'' in \emph{Proc. IEEE Conference on Computer Communications
  (INFOCOM)}, 2019, pp. 1387--1395.

\bibitem{Du2020SGquant}
Y.~{Du}, S.~{Yang}, and K.~{Huang}, ``High-dimensional stochastic gradient
  quantization for communication-efficient edge learning,'' \emph{IEEE
  Transactions on Signal Processing}, vol.~68, pp. 2128--2142, Mar. 2020.

\bibitem{Renzo2020SREJsac}
M.~D. {Renzo}, A.~{Zappone}, M.~{Debbah}, M.~{Alouini}, C.~{Yuen}, J.~D.
  {Rosny}, and S.~{Tretyakov}, ``Smart radio environments empowered by
  reconfigurable intelligent surfaces: How it works, state of research, and
  road ahead,'' \emph{IEEE Journal on Selected Areas in Communications}, pp.
  1--1, 2020.

\bibitem{Basar2019RIS}
E.~{Basar}, M.~{D. Renzo}, J.~{De Rosny}, M.~{Debbah}, M.~{Alouini}, and
  R.~{Zhang}, ``Wireless communications through reconfigurable intelligent
  surfaces,'' \emph{IEEE Access}, vol.~7, pp. 116\,753--116\,773, Aug. 2019.

\bibitem{Liaskos2018SCM}
C.~{Liaskos}, S.~{Nie}, A.~{Tsioliaridou}, A.~{Pitsillides}, S.~{Ioannidis},
  and I.~{Akyildiz}, ``A new wireless communication paradigm through
  software-controlled metasurfaces,'' \emph{IEEE Communications Magazine},
  vol.~56, no.~9, pp. 162--169, Sep. 2018.

\bibitem{Renzo2019SRE}
M.~D. Renzo, M.~Debbah, D.-T. Phan-Huy, A.~Zappone, M.-S. Alouini, C.~Yuen,
  V.~Sciancalepore, G.~C. Alexandropoulos, J.~Hoydis, H.~Gacanin, J.~d. Rosny,
  A.~Bounceur, G.~Lerosey, and M.~Fink, ``Smart radio environments empowered by
  reconfigurable {AI} meta-surfaces: an idea whose time has come,''
  \emph{EURASIP Journal on Wireless Communications and Networking}, vol. 2019,
  no.~1, p. 129, May 2019.

\bibitem{Liu2020RISchannelestimate}
H.~{Liu}, X.~{Yuan}, and Y.~A. {Zhang}, ``Matrix-calibration-based cascaded
  channel estimation for reconfigurable intelligent surface assisted multiuser
  {MIMO},'' \emph{IEEE Journal on Selected Areas in Communications}, pp. 1--1,
  2020.

\bibitem{QWu2019IRSBeamforming}
Q.~{Wu} and R.~{Zhang}, ``Intelligent reflecting surface enhanced wireless
  network via joint active and passive beamforming,'' \emph{IEEE Transactions
  on Wireless Communications}, vol.~18, no.~11, pp. 5394--5409, Nov. 2019.

\bibitem{Lin2020risarxiv}
S.~{Lin}, B.~{Zheng}, G.~C. {Alexandropoulos}, M.~{Wen}, M.~{D. Renzo}, and
  F.~{Chen}, ``{Reconfigurable Intelligent Surfaces with Reflection Pattern
  Modulation: Beamforming Design and Performance Analysis},'' \emph{arXiv
  e-prints}, p. arXiv:2008.02555, Aug. 2020.

\bibitem{Li2020RISUAV}
S.~{Li}, B.~{Duo}, X.~{Yuan}, Y.~{Liang}, and M.~{D. Renzo}, ``Reconfigurable
  intelligent surface assisted {UAV} communication: Joint trajectory design and
  passive beamforming,'' \emph{IEEE Wireless Communications Letters}, vol.~9,
  no.~5, pp. 716--720, May 2020.

\bibitem{Liu2020RISofdm}
S.~{Lin}, B.~{Zheng}, G.~C. {Alexandropoulos}, M.~{Wen}, F.~{Chen}, and
  S.~{Mumtaz}, ``Adaptive transmission for reconfigurable intelligent
  surface-assisted {OFDM} wireless communications,'' \emph{IEEE Journal on
  Selected Areas in Communications}, pp. 1--1, 2020.

\bibitem{Nadeem2020RIS}
Q.~{Nadeem}, A.~{Kammoun}, A.~{Chaaban}, M.~{Debbah}, and M.~{Alouini},
  ``Asymptotic max-min {SINR} analysis of reconfigurable intelligent surface
  assisted {MISO} systems,'' \emph{IEEE Transactions on Wireless
  Communications}, pp. 1--1, 2020.

\bibitem{Han2020IRSpowctrl}
H.~{Han}, J.~{Zhao}, D.~{Niyato}, M.~D. {Renzo}, and Q.~{Pham}, ``Intelligent
  reflecting surface aided network: Power control for physical-layer
  broadcasting,'' in \emph{Proc. IEEE International Conference on
  Communications (ICC)}, 2020, pp. 1--7.

\bibitem{Yang2020IRSnoma}
G.~{Yang}, X.~{Xu}, and Y.~{Liang}, ``Intelligent reflecting surface assisted
  non-orthogonal multiple access,'' in \emph{Proc. IEEE Wireless Communications
  and Networking Conference (WCNC)}, 2020, pp. 1--6.

\bibitem{Fu2019IRSnoma}
M.~{Fu}, Y.~{Zhou}, and Y.~{Shi}, ``Intelligent reflecting surface for downlink
  non-orthogonal multiple access networks,'' in \emph{Proc. IEEE Globecom
  Workshops (GC Wkshps)}, 2019, pp. 1--6.

\bibitem{Wu2020IRSswip}
Q.~{Wu} and R.~{Zhang}, ``Joint active and passive beamforming optimization for
  intelligent reflecting surface assisted {SWIPT} under {QoS} constraints,''
  \emph{IEEE Journal on Selected Areas in Communications}, pp. 1--1, 2020.

\bibitem{Guo2020sumrateRIS}
H.~{Guo}, Y.~{Liang}, J.~{Chen}, and E.~G. {Larsson}, ``Weighted sum-rate
  maximization for reconfigurable intelligent surface aided wireless
  networks,'' \emph{IEEE Transactions on Wireless Communications}, vol.~19,
  no.~5, pp. 3064--3076, May 2020.

\bibitem{Zhang2020LimitedPhase}
H.~{Zhang}, B.~{Di}, L.~{Song}, and Z.~{Han}, ``Reconfigurable intelligent
  surfaces assisted communications with limited phase shifts: How many phase
  shifts are enough?'' \emph{IEEE Transactions on Vehicular Technology},
  vol.~69, no.~4, pp. 4498--4502, Feb. 2020.

\bibitem{Di2020HybBF}
B.~{Di}, H.~{Zhang}, L.~{Song}, Y.~{Li}, Z.~{Han}, and H.~V. {Poor}, ``Hybrid
  beamforming for reconfigurable intelligent surface based multi-user
  communications: Achievable rates with limited discrete phase shifts,''
  \emph{IEEE Journal on Selected Areas in Communications}, vol.~38, no.~8, pp.
  1809--1822, Jun. 2020.

\bibitem{Yu2004waterfilling}
{W. Yu}, {W. Rhee}, S.~{Boyd}, and J.~M. {Cioffi}, ``Iterative water-filling
  for gaussian vector multiple-access channels,'' \emph{IEEE Transactions on
  Information Theory}, vol.~50, no.~1, pp. 145--152, Jan. 2004.

\bibitem{Li2015mmfairness}
Y.~{Li}, M.~{Sheng}, C.~W. {Tan}, Y.~{Zhang}, Y.~{Sun}, X.~{Wang}, Y.~{Shi},
  and J.~{Li}, ``Energy-efficient subcarrier assignment and power allocation in
  {OFDMA} systems with max-min fairness guarantees,'' \emph{IEEE Transactions
  on Communications}, vol.~63, no.~9, pp. 3183--3195, Sep. 2015.

\bibitem{Liu2020dataimportance}
D.~{Liu}, G.~{Zhu}, J.~{Zhang}, and K.~{Huang}, ``Data-importance aware user
  scheduling for communication-efficient edge machine learning,'' \emph{IEEE
  Transactions on Cognitive Communications and Networking}, pp. 1--1, 2020.

\bibitem{Shi2019RISEL}
S.~{Hua} and Y.~{Shi}, ``Reconfigurable intelligent surface for green edge
  inference in machine learning,'' in \emph{Proc. IEEE Globecom Workshops (GC
  Wkshps)}, 2019, pp. 1--6.

\bibitem{Yang2020MEC-ML}
B.~{Yang}, X.~{Cao}, X.~{Li}, Q.~{Zhang}, and L.~{Qian},
  ``Mobile-edge-computing-based hierarchical machine learning tasks
  distribution for {IIoT},'' \emph{IEEE Internet of Things Journal}, vol.~7,
  no.~3, pp. 2169--2180, Dec. 2019.

\bibitem{Wang2020LearningCentricICC}
S.~{Wang}, R.~{Wang}, Q.~{Hao}, Y.~{Wu}, and H.~V. {Poor}, ``Learning centric
  power allocation for edge intelligence,'' in \emph{Proc. IEEE International
  Conference on Communications (ICC)}, 2020, pp. 1--6.

\bibitem{johnson2018accuracypilotdata}
M.~Johnson, P.~Anderson, M.~Dras, and M.~Steedman, ``Predicting accuracy on
  large datasets from smaller pilot data,'' in \emph{Proc. The Annual Meeting
  of the Association for Computational Linguistics (ACL)}, Melbourne,
  Australia, Jul. 2018, pp. 450--455.

\bibitem{goldsmith_2005}
A.~Goldsmith, \emph{Wireless Communications}.\hskip 1em plus 0.5em minus
  0.4em\relax Cambridge University Press, 2005.

\bibitem{He2019raytracing}
D.~{He}, B.~{Ai}, K.~{Guan}, L.~{Wang}, Z.~{Zhong}, and T.~{Kürner}, ``The
  design and applications of high-performance ray-tracing simulation platform
  for {5G} and beyond wireless communications: A tutorial,'' \emph{IEEE
  Communications Surveys Tutorials}, vol.~21, no.~1, pp. 10--27, Firstquarter
  2019.

\bibitem{BELEITES2013samplesize}
C.~Beleites, U.~Neugebauer, T.~Bocklitz, C.~Krafft, and J.~Popp, ``Sample size
  planning for classification models,'' \emph{Analytica Chimica Acta}, vol.
  760, pp. 25 -- 33, Jan. 2013.

\bibitem{BMarks1978InnerApprox}
B.~R. Marks and G.~P. Wright, ``A general inner approximation algorithm for
  nonconvex mathematical programs,'' \emph{Operations Research}, vol.~26,
  no.~4, pp. 681--683, Aug. 1978.

\bibitem{Luo2010SDR}
Z.~{Luo}, W.~{Ma}, A.~M. {So}, Y.~{Ye}, and S.~{Zhang}, ``Semidefinite
  relaxation of quadratic optimization problems,'' \emph{IEEE Signal Processing
  Magazine}, vol.~27, no.~3, pp. 20--34, May 2010.

\bibitem{boyd2004convex}
S.~Boyd, S.~P. Boyd, and L.~Vandenberghe, \emph{Convex Optimization}.\hskip 1em
  plus 0.5em minus 0.4em\relax Cambridge university press, 2004.

\bibitem{Zhu2020AirComp}
G.~{Zhu}, Y.~{Wang}, and K.~{Huang}, ``Broadband analog aggregation for
  low-latency federated edge learning,'' \emph{IEEE Transactions on Wireless
  Communications}, vol.~19, no.~1, pp. 491--506, Jan. 2020.

\bibitem{BenTal2001cvx}
A.~Ben-Tal and A.~S. Nemirovskiaei, \emph{Lectures on Modern Convex
  Optimization: Analysis, Algorithms, and Engineering Applications}.\hskip 1em
  plus 0.5em minus 0.4em\relax USA: Society for Industrial and Applied
  Mathematics, 2001.

\bibitem{Yang2020MEC-MultiTaskLearning}
B.~{Yang}, X.~{Cao}, J.~{Bassey}, X.~{Li}, and L.~{Qian}, ``Computation
  offloading in multi-access edge computing: A multi-task learning approach,''
  \emph{IEEE Transactions on Mobile Computing}, pp. 1--1, Apr. 2020.

\bibitem{scikit-learn}
F.~Pedregosa, G.~Varoquaux, A.~Gramfort, V.~Michel, B.~Thirion, O.~Grisel,
  M.~Blondel, P.~Prettenhofer, R.~Weiss, V.~Dubourg, J.~Vanderplas, A.~Passos,
  D.~Cournapeau, M.~Brucher, M.~Perrot, and {{\'E}}douard Duchesnay,
  ``Scikit-learn: Machine learning in python,'' \emph{Journal of Machine
  Learning Research}, vol.~12, no.~85, pp. 2825--2830, Oct. 2011.

\bibitem{MNIST}
\BIBentryALTinterwordspacing
Y.~LeCun and C.~Cortes, ``{MNIST} handwritten digit database,'' 2010. [Online].
  Available: \url{http://yann.lecun.com/exdb/mnist/}
\BIBentrySTDinterwordspacing

\bibitem{Fashion-MNIST}
\BIBentryALTinterwordspacing
H.~Xiao, K.~Rasul, and R.~Vollgraf, ``Fashion-mnist: a novel image dataset for
  benchmarking machine learning algorithms,'' \emph{CoRR}, vol. abs/1708.07747,
  2017. [Online]. Available: \url{http://arxiv.org/abs/1708.07747}
\BIBentrySTDinterwordspacing

\bibitem{PointNet}
R.~Q. {Charles}, H.~{Su}, M.~{Kaichun}, and L.~J. {Guibas}, ``Pointnet: Deep
  learning on point sets for 3d classification and segmentation,'' in
  \emph{Proc. IEEE Conference on Computer Vision and Pattern Recognition
  (CVPR)}, 2017, pp. 77--85.

\bibitem{CARLA}
A.~Dosovitskiy, G.~Ros, F.~Codevilla, A.~Lopez, and V.~Koltun, ``{CARLA}: {An}
  open urban driving simulator,'' in \emph{Proc. The 1st Annual Conference on
  Robot Learning}, 2017, pp. 1--16.

\bibitem{Yan2018SECOND}
Y.~{Yan}, Y.~{Mao}, and B.~{Li}, ``{SECOND}: Sparsely embedded convolutional
  detection,'' \emph{Sensors}, vol.~18, no.~10, p. 3337, Oct. 2018.

\bibitem{Geiger2013IJRR}
A.~Geiger, P.~Lenz, C.~Stiller, and R.~Urtasun, ``Vision meets robotics: The
  kitti dataset,'' \emph{International Journal of Robotics Research}, vol.~32,
  pp. 1231--1237, Aug. 2013.

\end{thebibliography}

\end{document}